\documentclass[12pt]{article}
\usepackage[utf8]{inputenc}
\usepackage{amsmath,amssymb,mathtools,amsthm,fullpage}
\mathtoolsset{showonlyrefs}

\usepackage{listings}
\usepackage[linesnumbered, ruled]{algorithm2e}
\usepackage[toc]{appendix}
\usepackage{a4wide}
\allowdisplaybreaks
\makeatletter
\renewcommand{\@algocf@capt@plain}{above}% formerly {bottom}
\makeatother

\lstdefinelanguage{Julia}%
  {morekeywords={abstract,break,case,catch,const,continue,do,else,elseif,%
      end,export,false,for,function,immutable,import,importall,if,in,%
      macro,module,otherwise,quote,return,switch,true,try,type,typealias,%
      using,while},%
   sensitive=true,%
   alsoother={$},%
   morecomment=[l]\#,%
   morecomment=[n]{\#=}{=\#},%
   morestring=[s]{"}{"},%
   morestring=[m]{'}{'},%
}[keywords,comments,strings]%

\usepackage{bbm}
\usepackage[colorlinks,
			pdffitwindow=false,
      plainpages=false,
      pdfpagelabels=true,
     	pdfpagemode=UseOutlines,
      pdfpagelayout=SinglePage,
			bookmarks=false,
			colorlinks=true,
      hyperfootnotes=false,
			linkcolor=blue,
			citecolor=green!50!black]{hyperref}
\usepackage{bm}
\usepackage{dsfont}
\usepackage{ulem} % Strikethrough, but messes up \emph
\usepackage[makeroom]{cancel} % cancels math formulas nicely
\usepackage[pdftex,dvipsnames]{xcolor}  % Coloured text etc.
\definecolor{cadmiumgreen}{rgb}{0.0, 0.42, 0.24}
\definecolor{ao}{rgb}{0.0, 0.0, 1.0}

\newcommand{\mkadd}[1]{\textcolor{cadmiumgreen}{#1}}

\usepackage[english]{babel}

\usepackage[colorinlistoftodos,prependcaption,textsize=small]{todonotes}
\usepackage{xargs}                      % Use more than one optional parameter in a new commands
\usepackage{authblk}
\DeclareMathAlphabet\mathbfcal{OMS}{cmsy}{b}{n}

\title{Multi-set spectral clustering of time-evolving networks using the supra-Laplacian}
\author[1]{Gary Froyland}
\author[2,*]{Manu Kalia}
\author[3]{P\'eter Koltai}

\affil[1]{School of Mathematics and Statistics, University of New South Wales, Sydney NSW 2052, Australia}
\affil[2]{Department of Mathematics, Free University of Berlin, 14195 Berlin, Germany}
\affil[3]{Department of Mathematics, University of Bayreuth, 95440 Bayreuth, Germany}
\affil[*]{\url{m.kalia@fu-berlin.de}}

\date{\today}

\newtheorem{definition}{Definition}
\newtheorem{theorem}[definition]{Theorem}
\newtheorem{lemma}[definition]{Lemma}
\newtheorem{prop}[definition]{Proposition}

\newtheorem{example}{Example}

\newcommand*\sop[1]{\mathcal #1} % Static operator
\newcommand*\op[1]{\bm{\mathcal #1}} % Operator
\newcommand*\smat[1]{\mathrm{#1}} % Static matrix
\newcommand*\mat[1]{\bm{\mathrm{#1}}} % Matrix
\newcommand*\sfunc[1]{\lowercase{#1}} % Static function
\newcommand*\func[1]{\bm {\lowercase{#1}}} % function
\newcommand*\svect[1]{\mathrm{\lowercase{#1}}} % Static vector
\newcommand*\vect[1]{\bm {\mathrm {\lowercase{#1}}}} % Vector
\newcommand*\sset[1]{#1} % Static set
\newcommand*\set[1]{\bm #1} % Set
\newcommand*\spartn[1]{\mathfrak #1} % Static partition
\newcommand*\partn[1]{\bm{\mathfrak #1}} % Partition
\newcommand{\eval}{\bm \lambda} % Eigenvalue
\newcommand{\seval}{\lambda} % Static eigenvalue
\newcommand*\cfunc[1]{\bm{#1}} % Cheeger function

\newcommandx{\unsure}[2][1=]{\todo[linecolor=red,backgroundcolor=red!25,bordercolor=red,#1]{#2}}
\newcommandx{\imp}[2][1=]{\todo[linecolor=blue,backgroundcolor=blue!25,bordercolor=blue,#1]{#2}}
\newcommand*\adep[1]{#1^{(a)}}

\newcommand*\norm[1]{\overline{#1}}
\newcommand*\unnorm[1]{{#1}}
\newcommand*\comp[1]{#1^{\complement}}
\newcommand{\lap}{\sop L}
\newcommand{\slap}{\adep{\op L}}
\newcommand{\ulap}{\adep{\unnorm{\op L}}}
\newcommand{\nlap}{\norm{\adep{\op L}}}
\newcommand{\adj}{\sop W}
\newcommand{\sadj}{\adep{\op W}}

\newcommand{\ltemp}{\op{L}^{\rm temp}}
\newcommand{\ulspat}{\unnorm{\op{L}}^{\rm spat}}
\newcommand{\ultemp}{\unnorm{\op{L}}^{\rm temp}}
\newcommand{\nlspat}{\norm{\op{L}}^{\rm spat}}
\newcommand{\nltemp}{\norm{\op{L}}^{\rm temp}}
\newcommand{\adjspat}{\op{W}^{\rm spat}}
\newcommand{\adjtemp}{\op{W}^{\rm temp}}
\newcommand{\dlap}{\sop L^D}
\newcommand{\dadj}{\sop W^D}
\newcommand{\slapmat}{\adep{\mat L}}
\newcommand{\lspatmat}{\mat L^{\rm spat}}
\newcommand{\ltempmat}{\mat L^{\rm temp}}
\newcommand{\sadjmat}{\adep{\mat W}}
\newcommand{\adjspatmat}{\mat W^{\rm spat}}
\newcommand{\adjtempmat}{\mat W^{\rm temp}}
\newcommand{\ulapmat}{\adep {\unnorm {\mat L}}}

\renewcommand{\emph}[1]{\textit{#1}}

\begin{document}

\maketitle

\begin{abstract}
Complex time-varying networks are prominent models for a wide variety of spatiotemporal phenomena.
The functioning of networks depends crucially on their connectivity, yet reliable techniques for learning communities in time-evolving networks remain elusive.
We adapt successful spectral techniques from continuous-time dynamics on manifolds to the graph setting to fill this gap. 
We consider the supra-Laplacian for graphs and develop a spectral theory to underpin the corresponding algorithmic realisations.
We develop spectral clustering approaches for both multiplex and non-multiplex networks, based on the eigenvectors of the % inflated dynamic 
supra-Laplacian and specialised Sparse EigenBasis Approximation (SEBA) post-processing of these eigenvectors.
We demonstrate that our approach can outperform the Leiden algorithm applied both in spacetime and layer-by-layer, and we analyse voting data from the US senate (where senators come and go as congresses evolve) to quantify increasing polarisation in time.
\end{abstract}

\newpage
\section{Introduction}

Complex interconnected systems from diverse applications such as biology, economy, physical, political and social sciences can be modelled and analysed by networks~(\cite{newman2010networks,Barabasi2016}).
Network function is governed by its connectivity structure.
The learning of communities--the problem of finding groups of nodes in networks \cite{newman2010networks}--is thus a central issue when analysing networks, or their mathematical manifestations, graphs~(\cite{newman2004finding}). 
In many situations, this connectivity structure and thus the underlying communities vary over time. 
For example, one may consider an online social network composed of vertices represented by users, and edges that represent the degree of connection between users (likes, comments, post shares, tags etc.). Communities in such a network may correspond to mutual interests or those formed by friends and family. A temporal evolution of these communities is described by changes in user interaction, essentially strengthening or weakening the edges between users as time passes, see \cite{Greene2010, Bhat2015} for examples. 

Time-evolving networks have been introduced before in the literature as {\it multilayer} networks (\cite{newman2010networks, Kivela2014}),
% \mk{\sout{\cite{Kivela2014}}}), 
where the vertex set and edge weights may change over time.
A special case of a multilayer network is a {\it multiplex} network,  
where vertices are copied across layers and each vertex maintains temporal connections only with its counterparts one layer forward and backward in time. 
Multilayer and multiplex networks are used to model time-varying networks where each vertex behaves as a state varying in time, and are thus ubiquitous in the literature; see \cite{Magnani2021, Rossetti2018} for reviews on this topic. 
This study focuses on detecting communities in both multiplex networks and  multilayer networks (where the set of vertices themselves may change over time; see Sec.~\ref{sec:sen}). 

Community detection methods of multiplex networks include modularity maximisation adapted to spacetime graphs (\cite{Mucha2010, Greene2010}), random walk-based approaches (\cite{Kuncheva2015,Klus2022}), hierarchical clustering (\cite{Masuda2019}),  detect-and-track methods (\cite{Bhat2015}) and ensemble methods that generalise to multilayer networks (\cite{Tagarelli2017}). 
In this work we prioritise the development of novel techniques for Laplacians on undirected spacetime graphs to detect time-evolving spacetime communities, particularly in challenging situations  where existing non-spacetime and spacetime methods fail.
We do so for several reasons. First, the eigendata from graph Laplacians contain crucial information about the optimal partitioning of a graph, allowing one to relate their spectra and eigenvectors to balanced graph cuts and the resulting Cheeger constants from isoperimetric theory~(\cite{chung1996laplacians, AndersonJr85}). 
Second, graph Laplacians on multiplex networks, called {\it supra-Laplacians} (\cite{sole2013spectral,Radicchi2013,DeFord2017}), have structure that we exploit to  develop theoretical justifications for the quality of spacetime communities extracted from their eigendata.  
Third, spectral clustering is a well-studied and robust algorithm for community detection.
Spectral partitioning  is computationally efficient as the resulting supra-Laplacians are sparse and only require the computation of the first few eigenpairs.

Previous works have glossed over establishing results about the spectrum of supra-Laplacians (\cite{sole2013spectral, Radicchi2013, DeFord2017}). 
We develop such spectral theory, which also {formally quantifies} the performance of our spatiotemporal spectral partitioning through Cheeger constants.
From this new theory we formulate an algorithm to detect multiple spatiotemporal clusters; our theory provides a   quality guarantee of the partitions across time.
Finally, we adapt our spectral partitioning approach to non-multiplex type networks, where vertices may appear and disappear in time.
Our {spectral clustering constructions  determine \textit{several spatiotemporal clusters}, including in \textit{multilayer graphs} where vertices may come and go.
Each of} these aspects have to our knowledge not been formulated before.

Our constructions are inspired by the {\it dynamic} and {\it inflated dynamic} Laplacians formulated for continuous-space dynamics on manifolds (\cite{Froyland2015}) and spacetime manifolds (\cite{FrKo23,AFK24}), respectively. 
The dynamic Laplacian for graphs was formulated in \cite{FroylandKwok2015}. 
The supra-Laplacian is the graph analogue of the inflated dynamic Laplacian, where inflation refers to expanding the graph by copying vertices across time. 
We will use the terminology ``supra-Laplacian'' in the sequel, to maintain consistency with established literature~(\cite{sole2013spectral,Radicchi2013,DeFord2017}).

Our main contributions are as follows:
\begin{enumerate}
    \item We construct the unnormalised and normalised supra-Laplacian on graphs and state the corresponding spacetime Cheeger inequalities, which provide worst-case guarantees on clustering quality. In Prop.~\ref{prop:h}, we also prove an intuitive result concerning the monotonicity of spacetime Cheeger constants {with respect to} the number of spacetime partition elements.
    \item Vertices in individual network layers are typically connected across time with certain strengths, which {represent} diffusion across time {in the associated supra-Laplacian}. We analyse the limit of increasing temporal connectivity strength, and in Thm.~\ref{thm:var} {(resp.\ Thm.~\ref{thm:normL})} we prove that the spectra of the unnormalised {(resp.\  normalised)} supra-Laplacians approach that of the dynamic {Laplacian (resp.\ normalised temporal Laplacian)} in {the} hyperdiffusion limit. As a result, we gain a better understanding of how to select the appropriate diffusion constant for the supra-Laplacian for graphs, which is important for the spectral partitioning algorithm.
    \item We formulate a novel spectral partitioning algorithm {(Algorithm~\ref{alg:spec_part})} in conjunction with the Sparse Eigenbasis Algorithm (SEBA) (\cite{Froyland2019}) for  multiplex networks to find \textit{several} spatiotemporal clusters from \textit{several} spacetime eigenvectors. We illustrate the efficacy of our approach in the challenging setting of small, slowly varying networks, and demonstrate superior time-varying cluster detection when compared to (i) a spacetime  application of the Leiden algorithm~(\cite{Traag2019}) and (ii) Leiden applied on individual time slices.
    \item In Algorithm \ref{alg:spec_part_nonmultiplex} we formulate a spacetime clustering procedure for  non-multiplex networks, which are characterised by time-varying edge weights, and vertices that appear and disappear in time. We analyse a real-world network of voting similarities between US senators in the years 1987--{2025}, {previously studied by e.g.\ }(\cite{Mucha2010,Waugh2009}), and show increasing political polarisation over recent decades.
\end{enumerate}

An outline of the paper is as follows.
We begin with a simple motivating example of a multiplex network with 5 vertices and 5 time slices in Sec.~\ref{sec:motivatingex} that captures many of the fundamental ideas of our approach. We formulate the main spacetime constructions in
Sec.~\ref{sec:spacetime}. 
In Sec.~\ref{sec:eigenproblem} we analyse the spectral structure of the supra-Laplacian and prove spectral results in the hyperdiffusion limit. 
Our spectral partitioning algorithm is formulated in Sec.~\ref{sec:matrix_alg}.   
In Sec.~\ref{sec:example_networks} we demonstrate the algorithm on two model networks. Finally in Sec.~\ref{sec:sen} we adapt the spectral partitioning algorithm to non-multiplex networks and apply it to detect communities in a real world network of voting similarities in the US Senate. A Julia implementation of the algorithms and the examples presented here can be found at \url{https://github.com/mkalia94/TemporalNetworks.jl}.

\subsection{Motivating example}
\label{sec:motivatingex} 

The first purpose of this example is to highlight how a spatiotemporal Laplacian construction is able to identify evolving clusters together with the times of their (partial) appearance and disappearance in evolving networks. Its second purpose is to informally introduce the main objects that support our theoretic and algorithmic constructions.
 
Consider the spacetime network as presented in Fig.~\ref{fig:example}. 
\begin{figure}[h!]
    \centering  \includegraphics[width=\textwidth]{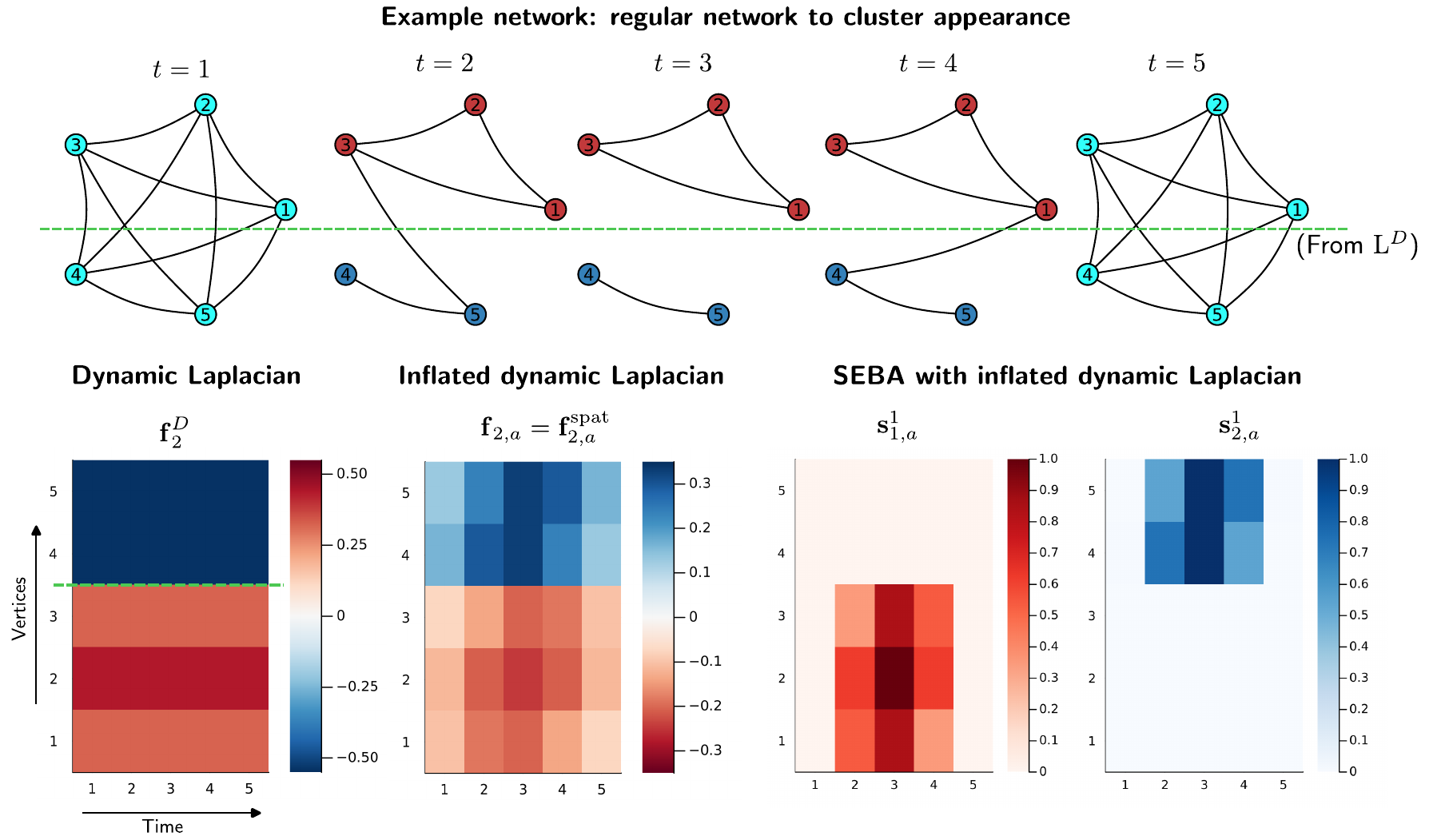}
    \caption{ Comparing spectral partitioning methods to discover communities in spacetime graphs.
    {\it Upper}:  spacetime graph/multiplex network composed of 5 time slices with 5 vertices each. As time progresses, {clusters composed of vertices $\{1,2,3\}$  and $\{4,5\}$ appear} from a 4-regular graph before {vanishing}.
    {\it Lower left:} The second eigenvector $\svect{f}^D_2 = [0.324, 0.442, 0.324, -0.545, -0.545]^\top$ of the dynamic Laplacian (copied across time to create a spacetime vector $\vect{F}_2^D$)  reveals a partition, demarcated by green dashed line, and is composed of vertices $\{1,2,3\}$ and $\{4,5\}$ respectively.
    This partition remains the same for all time, and does not reveal the appearing and disappearing of the clusters in the network.  {\it Lower right:} The SEBA algorithm (Algorithm~\ref{alg:spec_part} - step 5) applied to the second eigenvector $\vect{F}_{2,a}$ (in this case equal to the first spatial eigenvector $\vect{F}^{\rm spat}_{2,a}$) of the supra-Laplacian yields two vectors $ {\vect S}_{1,a}^{1}$ and ${\vect S}_{2,a}^{1}$, {which} reveal the emergence of two disjoint clusters at time $t=2$ and their disappearance after $t=4$. 
    Note that the 4-regular graphs at times $t=1$ and $t=5$ are not partitioned because they contain no clusters. }
    \label{fig:example}
\end{figure} 
The network is composed of 5 vertices for each of the 5 time slices. Each edge has a weight that switches between zero and one arbitrarily as time progresses. At time %$t=3$
$t=2$, two communities (composed of the set of vertices $\{1,2,3\}$ and $\{4,5\}$) appear, then disappear, returning to the regular graph at $t=5$. The goal is then to discover the communities represented by vertices $\{1,2,3\}$ and $\{4,5\}$ at times $t=2,3,4$. 
Let $\sset W = \{\smat W^{(1)}, \dots, \smat W^{(5)}\}$  be the collection of adjacency matrices for each time slice; see the black lines in Fig.~\ref{fig:example} (upper).
 
\textbf{Persistent communities:} In \cite{FroylandKwok2015} the {\it dynamic Laplacian} for graphs subjected to vertex permutations was introduced as a method to detect \textit{persistent} communities in spacetime graphs. 
Define the average adjacency $\smat W^D := \frac{1}{5}\sum_i \smat W^{(i)}$.
The dynamic Laplacian $\smat L^D$ is the (unnormalised) Laplacian for the weights $\smat W^D$.
Spectral partitioning is then performed on $\smat L^D$. The first eigenvector of $\smat L^D$ corresponding to eigenvalue $0$ is simply $\svect f^D_1 = [1, \dots, 1]^\top$. Using a zero threshold on the second eigenvector $\svect f^D_2 = [0.324, \ 0.442, \ 0.324, \ -0.545, \ -0.545]^\top$ reveals the underlying partition composed of vertices $\{1,2,3\}$ and $\{4,5\}$, which is indeed the optimal fixed-in-time cut for this spacetime graph (green dashed line in Fig.~\ref{fig:example} (lower left)).

\textbf{Evanescent communities:} A fixed-in-time partition cannot reveal the full temporal connectivity information of the spacetime graph. 
We seek communities that are present for substantial portions of time, but \textit{do not necessarily persist through the entire time duration}. To capture such transient communities we use the graph Laplacian constructed on a spacetime graph with temporal diffusion along the temporal arcs. In the literature, such a construction has been called a \emph{supra-Laplacian}~(\cite{Gmez2013,sole2013spectral}).
We discuss its strong relations to the dynamic Laplacian in Thm.~\ref{thm:var}.

\begin{figure}[htbp]
\centering
\includegraphics[width=\textwidth]{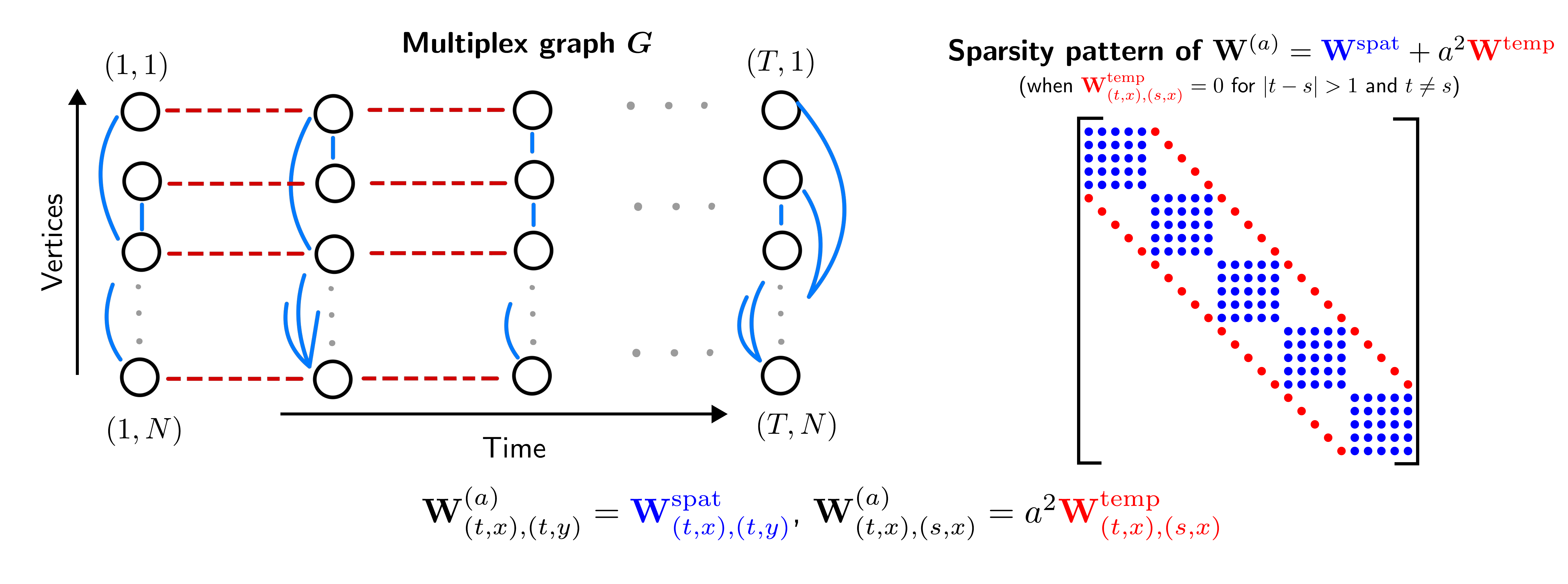}
\caption{The multiplex network framework. (Left) A spacetime graph $\set{G}$ with time/layers in the horizontal direction and vertices {within each time slice/layer} in the vertical direction. The connections form the adjacency  $\sadjmat$ which can be split into {blue} spatial ($\adjspatmat$) and {red} temporal ($a^2 \adjtempmat$) components respectively. (Right) Sparsity pattern of the matrix representation of {$\sadjmat$ with respect to the ordering \eqref{eq:lexicographic_ordering}.}}
\label{fig:intro}
\end{figure}

A simple example of temporal diffusion is diffusion to nearest neighbours in time. 
The corresponding \textit{temporal Laplacian} is  
\begin{equation}
    \smat L^{\mathrm{temp}} = \left[
\begin{array}{ccccc}
1 & -1 & 0 & 0 & 0 \\
-1 & 2 & -1 & 0 & 0 \\
0 & -1 & 2 & -1 & 0 \\
0 & 0 & -1 & 2 & -1 \\
0 & 0 & 0 & -1 & 1 \\
\end{array}
\right].
\end{equation}
Defining $\smat L^{(i)}$ to be the (unnormalised) Laplacian for the weights $\smat W^{(i)}$, we {denote} by $\slapmat$ the supra-Laplacian of the spacetime graph shown in Fig.~\ref{fig:example} (upper), where the boldface indicates the spacetime nature of the matrix. The supra-Laplacian $\slapmat$ is given by
\begin{equation}
    \slapmat = \mathrm{blockdiag}(\smat L^{(1)},\ldots,\smat L
    ^{(5)}) + a^2(\smat L^{\rm temp} \otimes \smat I),
\end{equation}
where the parameter $a > 0$ weights the contributions of the temporal component relative to the spatial components.
Here, $\otimes$ denotes the Kronecker product of matrices (cf.~\eqref{eq:Kronecker}) and $\smat I$ the identity matrix. With the ``space-first'' ordering of the multiplex (spacetime) network's vertices as in Fig.~\ref{fig:intro} (upper left), the supra-Laplacian $\slapmat$ and its weight matrix $\sadjmat$ both have the sparsity structure shown in Fig.~\ref{fig:intro} (upper right). In this example we choose $a=2$ to appropriately scale spatial and temporal components (see Algorithm~\ref{alg:spec_part} for details). 
We now perform spectral partitioning on the supra-Laplacian. As in the dynamic case, the first eigenvector is simply $\vect F_{1,a} = [1, \dots, 1]^\top$. The second eigenvector $\vect F_{2,a}$ indicates the spacetime partition. 
To create the spacetime partition
we apply the SEBA algorithm (\cite{Froyland2019}) to two vectors constructed from $\vect F_{2,a}$ (see Algorithm~\ref{alg:spec_part}) which returns two vectors $\smash{\vect S_{1,a}^{1}}$ and $\smash{\vect S_{2,a}^{1}}$ %$\smash{\vect S_{(1,2),a}^{1}}$ 
containing the clusters $\{1,2,3\}$ and $\{4,5\}$ separately, at times $t=2,3,4$, seen in Fig.~\ref{fig:example} (lower left). This example shows that appearance and disappearance of communities (in time) can be extracted by spacetime spectral computations.

\section{Spacetime graphs}\label{sec:spacetime}

The spectral study of time-evolving networks requires extensive notation.
To help the reader navigate this notation we list some typical examples in Table~\ref{tab:notation}. 
Time-evolving networks have nodes indexed by both space and time. Generally, we distinguish spacetime objects from purely spatial or purely temporal objects by denoting the former by boldface symbols. Further, operators are calligraphic, matrices are upper case upright (roman), functions are italic, vectors are lower case upright, sets are upper case italic, and partitions of sets are fraktur letters. 
Scalars, indices, nodes in graphs, and other objects have  a standard notation that is not synchronised with the previous choices, but their identification should be clear from the context. 
For instance, $a$ will always be a scalar and never a function.
\begin{table*}[h]
    \centering
    \begin{tabular}{l|c|c}
      object type   & space or time object& spacetime object \\ \hline\hline
        operator & $\sop{W}, \sop{L}$ & $\op W, \op L$ \\
        matrix & $\smat W^{(i)}, \smat L^D$ & $\slapmat$ \\
        function & $\sfunc f^D$ & $\func f_{2,a}$ \\
        vector & $\svect{f}^{\mathrm{temp}}$ & $\vect f_{2,a}$ \\
        set & $\sset V$,  $\sset V'$ & $\set X_k$ \\
        partition & $\spartn{X}$ & $\partn X$ \\
        Cheeger 
        constant & $h_K$ & $\cfunc h_K$ \\
        eigenvalues & $\seval_k$ & $\eval_k$
    \end{tabular}
    \caption{Reference table with examples of notation.}
    \label{tab:notation}
\end{table*}
We distinguish operators and functions on graphs from matrices and vectors representing them, both to align with existing literature and to draw the line between an intuitive, intrinsic description of objects (in terms of vertices and edges) and their representations for algorithms that use a specific enumeration of the vertices.

\subsection{Spacetime graphs and clusters}
Let $G = (\sset V,\sset E, \adj)$ be a {general undirected} graph with finite sets of vertices $\sset V \subset \mathbb{N}$ and edges $(x,y) \in \sset E$ for vertices $x, y \in \sset V$.  The function $\adj: \sset V \times \sset V \to  \mathbb{R}_0^+$ assigns weights to the edges such that $\adj (x,y)> 0 $ and is symmetric in arguments. The degree $\sfunc d(x)$ of a vertex $x$ is given by $\sfunc d(x)= \sum_y\adj(x,y)$.

We now extend these definitions to spacetime graphs. Let $\sset V = \{1, \dots, N\}$ denote the spatial vertex set and $\sset V' = \{1, \dots, T\}$ be the temporal vertex set for finite $N, T \in \mathbb N$. We define the spacetime vertex set $\set V= \sset V' \times \sset V$ which gives rise to the \textit{spacetime, multiplex graph} $\set{G}$ defined by the tuple $\set{G} = (\set{V},\set{E}, \sadj)$ where $\set{E}$  is the edge set connecting vertices $\set{V}$ defined as follows,
\begin{align*}
    \set E &= \{ ((t,x),(s,y)) \ | \ (t,x) \sim (s,y) \ \textnormal{in } \set{G} \}.
\end{align*}
The notation $(t,x) \sim (s,y)$ denotes a connection between vertices $(t,x)$ and $(s,y)$ in $\set{G}$. 
 For every graph $\set{G}$ we associate a weight function $\sadj$ parameterised by $a>0$  such that $\smash{ \sadj_{(t,x),(s,y)}\in\mathbb{R}^+_0 }$ is the weight of the edge joining the vertices $(t,x)$ and $(s,y)$.  We decompose $\sadj$ into $\adjspat$ and $a^2\adjtemp$,  the spatial and temporal components of $\sadj$, so that
\begin{equation}
    \sadj = \adjspat + a^2\adjtemp,
    \label{eq:supraW}
\end{equation}
and
\begin{align*}
    \adjspat_{(t,x),(s,y)} &= 0, \  t\neq s,  \\
    \adjtemp_{(t,x),(s,y)} &= 0, \ x \neq  y. 
\label{eq:supraSplit}\stepcounter{equation}\tag{\theequation}
\end{align*}
We further restrict $\adjtemp$ so that it is space-independent, i.e.,
\begin{equation}
    \adjtemp_{(t,x),(s,x)} = \adjtemp_{(t,y),(s,y)} =: \sop W'_{t,s} \textnormal{ for all $t,s,x,y$}.
    \label{eq:Wtemp_op}
\end{equation}

\paragraph{{Clustering} spacetime graphs.} 

Individual clusters are of spacetime type and can thus {appear or disappear, or eject or absorb vertices} as time passes.  For clarity, we {formally} define these transitions below.  {When discussing clustering we will} use the notion of \textit{packing} which is more general than \textit{partitioning}. 
A $K$-packing $\partn X = \{\set{X}_1,\ldots,\set{X}_K\}$ of the vertex set $\set{V}$ is defined by a collection of disjoint subsets $\set X_i \subset \set V$, $1 \leq i \leq K$. The {$K$-packing $\partn{X}$ can  be one of two possible types: (i) \textit{fully clustered}, where $\partn{X}$ partitions $\set{V}$ or (ii) \textit{partially clustered}, where $\bigcup_{k=1}^K\set{X}_k\subsetneq \set{V}$. In this latter case, we denote the unclustered vertices by $\set \Omega=\set{V}\setminus \bigcup_{k=1}^K\set{X}_k$.}

\begin{example}
\label{ex:example}
In Fig.~\ref{fig:example}, {the 3-packing} $\partn{X}=\{\set\Omega,\set{X}_1,\set{X}_2\}$ {partitions the 25-member spacetime vertex set $\set V$}, where $$
\set\Omega=\{(1,1),(1,2),(1,3),(1,4),(1,5),(5,1),(5,2),(5,3),(5,4),(5,5)\},$$ $$\set{X}_1=\{(2,1),(2,2),(2,3),(3,1),(3,2),(3,3),(4,1),(4,2),(4,3)\},$$ 
$$\mbox{and }\set{X}_2=\{(2,4),(2,5),(3,4),(3,5),(4,4),(4,5)\}.$$
\end{example}

\setcounter{example}{0}

The nontrivial evolution of spacetime sets can be described according to definitions below.
\begin{definition}
    \label{defn:graphevolve}
    \,
\begin{enumerate}
\item 
Let 
$\partn J\subset \partn X$ be the minimal collection of spacetime sets such that 
    \begin{equation}
    \label{eq:shrinkdisappeardef}
      \emptyset\neq  \left\{ x\in V \,:\, (t,x) \in \set{X} \right\} \subseteq \left\{ x\in \sset V \,:\, (t+1,x) \in \set{X} \cup \left(\textstyle{\bigcup_{\set{Y}\in \partn{J}}} \set{Y}\right) \right\}.
    \end{equation}
    If $\partn J\neq \emptyset$ we say that the spacetime set $\set{X}$ \emph{ejects} into $\partn J$ from time $t$ to $t+1$.
    In words, the space vertices in $\set{X}$ at time $t$ are partially redistributed to {all} packing elements forming the collection $\partn J$ at time $t+1$. 
    
 In the special case where $\left\{ x\in \sset V \,:\, (t+1,x) \in \set{X}\right\}=\emptyset$, then $\set X$ is said to \emph{disappear into  $\partn J$ at time $t+1$}.
In words, the space vertices in $\set{X}$ at time $t$ are fully redistributed to {all} packing elements forming the collection $\partn J$ at time~$t+1$.

\item Let $\partn J\subset \partn X$ be the minimal collection of spacetime sets such that 
    \begin{equation}
    \label{eq:growappeardef}
        \left\{ x\in \sset V \,:\, (t,x) \in \set{X} \cup \left(\textstyle{\bigcup_{\set{Y}\in \partn{J}}} \set{Y}\right) \right\}\supseteq\left\{ x\in V \,:\, (t+1,x) \in \set{X} \right\}\neq\emptyset.
         \end{equation}
        If $\partn J\neq \emptyset$ we say that the spacetime set $\set{X}$ \emph{absorbs from $\partn J$ from time $t$ to $t+1$}.  
    In words, the space vertices in $\set{X}$ at time $t+1$ include vertices from {all} other packing elements forming the collection $\partn J$ at time~$t$. 
    
 In the special case where $\left\{ x\in \sset V \,:\, (t,x) \in \set{X}\right\}=\emptyset$, then $\set X$ is said to \emph{appear from $\partn J$ at time $t+1$}.
In words, the space vertices in $\set{X}$ at time $t+1$ are drawn solely from all packing elements forming the collection $\partn J$ at time~$t$.
\end{enumerate}

{Note that if part 1 of Definition \ref{defn:graphevolve} applies across times $t$ and $t+1$, so too does part 2. By convention in this paper we  select the part of Definition \ref{defn:graphevolve} that yields an $\set X$ with the largest cardinality.}

\end{definition}

\begin{figure}[htbp]
\centering
\includegraphics[width=\textwidth]{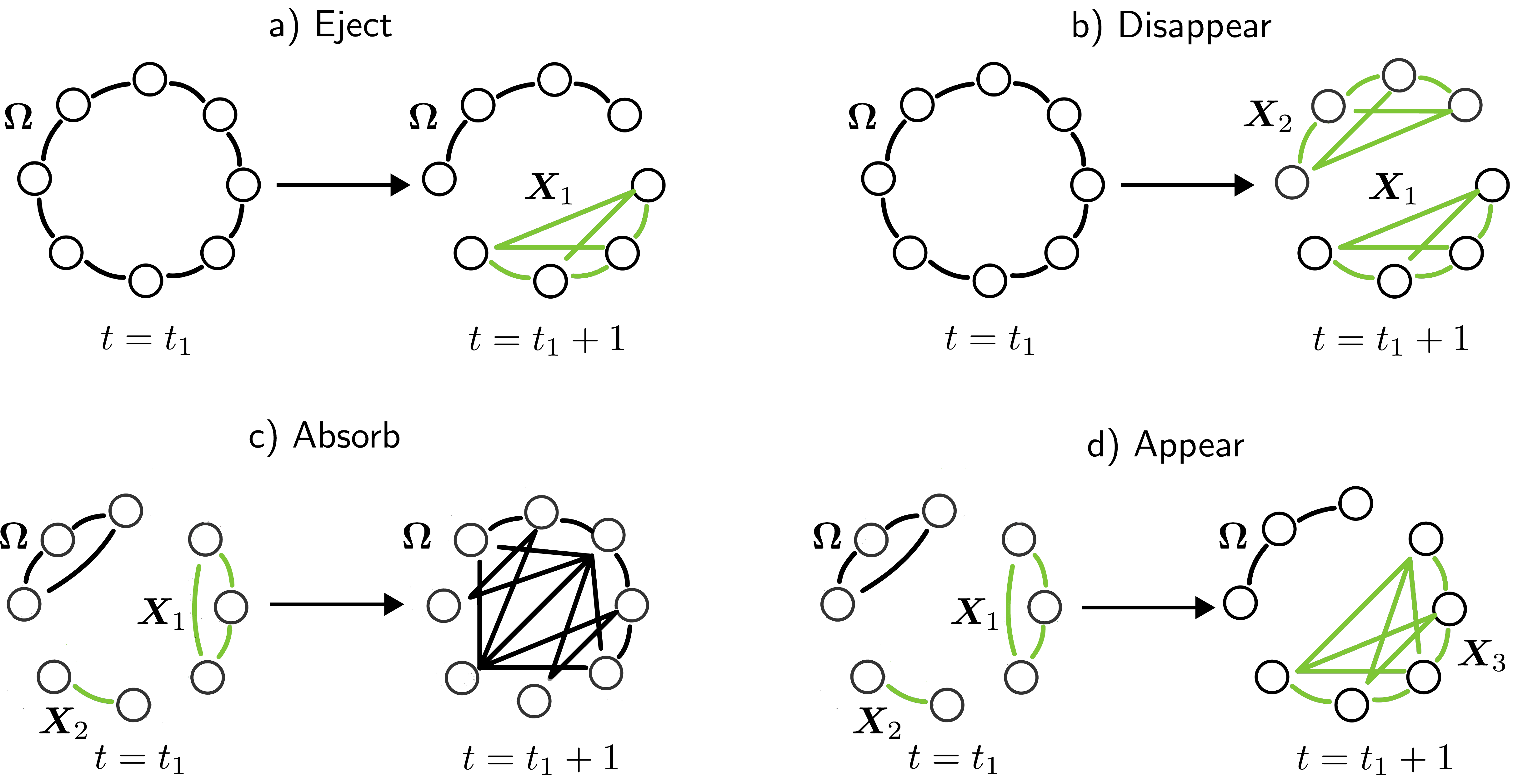}
\caption{
Upper Left: $\set \Omega$ ejects into $\{\set X_1\}$.  Upper Right: {$\set \Omega$ disappears into $\{\set X_1, \set X_2\}$.}   Bottom Left: {$\set \Omega$ absorbs from $\{ \set X_1, \set X_2\}$.}  Bottom right:   {$\set X_3$ appears from $\{\set X_1, \set X_2\}$.}}  
\label{fig:clustertransition}
\end{figure}

\begin{example}[cont...]
    \label{ex:example2}
Using the explicit element listings for $\set\Omega,\set{X}_1, \set{X}_2$ in Example \ref{ex:example}, we see that in Fig.~\ref{fig:example} {the single-element collection $\{\set \Omega\}$ disappears into $\{\set X_1,\set X_2\}$   at time $t=1$, and the single-element collection \{$\set \Omega\}$ appears from $\{\set X_1, \set X_2\}$ at time $t=5$.}
\end{example}

\subsection{Laplacians of spacetime graphs}
\label{sec:LapSpacetime}

From \cite{chung1996lectures}, the unnormalised and normalised graph Laplacians for general graphs are defined as follows. The unnormalised graph Laplacian $\lap$ acts on functions $\sfunc f: \sset V \to \mathbb{R}$ and is given by
\begin{equation}
    \lap \sfunc f(x) = \sum_y \adj(x,y) \left(\sfunc f(x)-\sfunc f(y)\right).
    \label{eq:Laplace}
\end{equation}
The corresponding normalised Laplacian $\norm \lap$ is given by
\begin{equation}
    \norm\lap \sfunc f(x) = \sum_y \adj(x,y) \left(  \frac{\sfunc f(x)}{\sfunc d(x)}-\frac{\sfunc f(y)}{\sfunc d(x)^{1/2}\sfunc d(y)^{1/2}}\right).
\end{equation}

We extend these definitions to the spacetime case. First, we define by  $\adep{\func d}$, $ \func d^{\rm spat}$ and $\func d^{\rm temp}$  the degree functions over vertices $(t,x) \in \set V$ and subsets of vertices~$\set X \subset \set {V}$. 
The expressions are as follows:
\begin{align*}
       \adep{\func d}(t,x) &= \sum_{(s,y)\in\set{V}}\sadj_{(t,x),(s,y)},   
         &\adep {\func {d}}(\set{X}) &= \sum_{(t,x) \in \set {X}} \adep{\func d}(t,x), \\
       \func d^{\rm spat}(t,x) &= \sum_{y \in \set V} \adjspat_{(t,x),(t,y)}, 
        &{\func d^{\rm spat}}(\set{X}) &= \sum_{(t,x) \in \set{X}} \func d^{\rm spat}(t,x), \\
       \func d^{\rm temp}(t,x) &= \sum_{s \in \sset V'} \adjtemp_{(t,x),(s,y)},  &{\func d^{\rm temp}}(\set X) &= \sum_{(t,x) \in \set X} \func d^{\rm temp}(t,x).
       \stepcounter{equation}\tag{\theequation}\label{eq:deg} 
\end{align*}
Note that by \eqref{eq:Wtemp_op}, $\func d^{\rm temp}(t,x) = \sum_{s \in \sset V'}\sop W'_{t,s}$ and is therefore independent of $x$. Thus we use $\sfunc d^{\rm temp}(t)$ in the rest of the paper.
On the graph $\set{G}$ equipped with weights $\sadj$, 
we define the  unnormalised graph supra-Laplacian $\slap$ as an operator acting on functions $\func F:\set{V} \to  \mathbb{R}$ as
\begin{equation}
    \slap \func F(t,x) := \sum_{(s,y) \in \set V}\sadj_{(t,x),(s,y)}\left(\func F(t,x) - \func F(s,y)\right).
    \label{eq:generalLap}
\end{equation}
We define the normalised graph supra-Laplacian $\nlap$ acting on functions $\func F:\set V \to \mathbb R$ as
\begin{equation}
    \nlap \func F(t,x) := \sum_{(s,y) \in \set V}\sadj_{(t,x),(s,y)}\left(\frac{\func F(t,x)}{\adep {\func d}(t,x)} - \frac{\func F(s,y)}{\adep {\func d}(t,x)^{1/2}\adep {\func d}(s,y)^{1/2}}\right).
\end{equation}
We note that in the case of the normalised supra-Laplacian, the corresponding graph $\set{G}$ should not contain isolated vertices in order to make sure~$\adep {\func d} (t,x)>0$, $\forall (t,x) \in \set V$. 
As the operator $\ulap$ is linear in the argument $\sadj$, we decompose $\ulap$ similarly to $\sadj$ in \eqref{eq:supraW} as,
\begin{equation}
    \slap = \ulspat + a^2 \ultemp,
\end{equation}
where
\begin{equation}
    \begin{aligned}
        \ulspat \func F(t,x) &= \sum_{y \in \set V}\adjspat_{(t,x),(t,y)}\left(\func F(t,x)-\func F(t,y)\right),\\
        \ultemp \func F(t,x) &= \sum_{s \in \sset V'}\adjtemp_{(t,x),(s,x)} \left(\func F(t,x)-\func F(s,x)\right).
    \end{aligned}
    \label{eq:lspat_ltemp}
\end{equation}
 
For now it is sufficient to note that by construction all graph Laplacians considered here are self-adjoint (with respect to different inner products), positive semidefinite, and hence have purely real nonnegative spectrum. The standard inner product for real vectors (or equivalently, for functions) $f,g$ will be denoted by~$\langle f,g\rangle$. This inner product notation will be used in several scenarios, but to which spaces $f$ and $g$ belong will always be clear from the context.

\begin{definition}
    For $a \ge 0$ and $k=1,2,\ldots, TN$, we denote the $k$-th eigenpair (in ascending order) of $\ulap$ by $(\unnorm \eval_{k,a}, \unnorm {\func F}_{k,a})$ and the $k$-th eigenpair of $\smash{ \nlap }$ by $(\norm \eval_{k,a}, \norm {\func F}_{k,a})$.
    \label{def:eigenpair}
\end{definition}

Finally we mention the \emph{dynamic Laplacian} (\cite{FroylandKwok2015}) $\dlap$ acting on functions $\sfunc f: \sset V\to\mathbb{R}$.
It is derived by first averaging the spatial adjacency $\adjspat$ over the time fibers to obtain $\dadj$, indexed as follows:
\begin{equation}
    \dadj_{x,y} :=  \frac{1}{T}\sum_{t \in \sset V'}  \adjspat_{(t,x),(s,y)} = \frac{1}{T}\sum_{t \in \sset V'}  \sadj_{(t,x),(t,y)}.
    \label{eq:AvgW}
\end{equation}
Consider the graph $\sset{G}^D = (\sset V, \sset E, \dadj)$ with edges {from the} set $\sset E$ connecting vertices in $\sset V$ and define the average degree $\sfunc d^D(x)$ by
\begin{equation}
    {\sfunc d^D}(x) := \frac{1}{T}\sum_{t \in \sset V'} \func d^{\rm spat}(t,x). 
\end{equation}
Then the unnormalised dynamic Laplacian $\dlap$ acts over functions $\sfunc f: \sset V \to \mathbb{R}$ and is given by
\begin{equation}
    \label{eq:dLap}
    \dlap \sfunc f(x) := \sum_{y \in \sset V}\dadj_{x,y} \left( \sfunc f(x) - \sfunc f(y) \right).
\end{equation}

\subsection{Graph cuts in spacetime}
\label{ssec:spacetime_cuts}

We define the (edge) cut value $\func \sigma(\set{X})$ between a subset $\set{X} \subset \set{V}$ and its complement $\comp{\set X} := \set V \setminus \set X $ by
\begin{equation}
     \func \sigma(\set {X}) :=  \sum_{ (t,x) \in \set{X}} \sum_{(s,y) \in \comp{\set{X}}} \sadj_{(t,x),(s,y)} \,.
     \label{eq:cut}
\end{equation}
We think of clusters being ``good'' if their cut values are low.
For the  graph $\set{G} = (\set{V},\set{E}, \sadj)$, we are interested in the {\it balanced graph cut problem}. We consider a $K$-packing $\partn{X} = \{\set{X}_1, \dots ,\set{X}_K\}$ of the spacetime vertex set $\set{V}$. Recall that
% , such that 
$\set{X}_k \cap \set{X}_l = \emptyset$, $\forall k \neq l$. The goal is to minimise  $\max _{k \in \{1, \dots, K\}}\func \sigma(\set{X}_k)$ produced by such a packing while making sure that the node count or degree of each individual $\set X_k$ {remains large}. The standard quantity to minimise over the $K$-packing $\partn{X}$ are the unnormalised and normalised Cheeger ratios $\cfunc  H(\set X_k)$ and $\norm {\cfunc H}(\set X_k)$ defined by
\begin{equation}
     \cfunc H(\set X) :=   \frac{\cfunc  \sigma(\set{X})}{|{\set{X}}| }, \quad \norm {\cfunc H}(\set X) := \frac{\func \sigma(\set{X})}{\adep {\func d}(\set{X}) }. 
     \label{eq:Cheeger}
\end{equation}
Denotethe maximal Cheeger ratio of $\partn{X}$ by $\cfunc H_K(\partn{X}):=\max_{1\le k\le K} \cfunc H(\set{X}_k)$.
The unnormalised Cheeger constant $\cfunc h_K$ is given by 
\begin{equation}
    \begin{aligned}
         \cfunc h_K := \min_{\partn{X}\text{ is a $K$-packing of }\set{V}} \cfunc H_K(\partn{X}) &= \min_{\set{X}_1 , \dots ,\set{X}_K} \max_{1\le k\le K} \cfunc H(\set{X}_k) \\
         &= \min_{\set{X}_1, \dots ,\set{X}_K} \max_{1\le k\le K} \frac{\func \sigma(\set{X}_k)}{| {\set{X}_k} | }.
    \end{aligned}
     \label{eq:cheegerconstant}
\end{equation}
Similarly, {define} $\norm {\cfunc H}_K(\partn X) := \max_{1 \le k \le K} \norm {\cfunc H}(\set X_k)$. The normalised Cheeger constant $\norm {\cfunc h}_K$ is given by
\begin{equation}
\begin{aligned}
         \norm {\cfunc h}_K := \min_{\partn{X}\text{ is a $K$-packing of }\set{V}} \norm {\cfunc H}_K(\partn{X}) &= \min_{\set{X}_1 , \dots ,\set{X}_K} \max_{1\le k\le K} \norm {\cfunc H}(\set{X}_k) \\
         &= \min_{\set{X}_1, \dots ,\set{X}_K} \max_{1\le k\le K} \frac{\cfunc \sigma(\set{X}_k)}{\adep {\func d} ({\set{X}_k}) }.
         \end{aligned}
     \label{eq:cheegerconstant_norm}
\end{equation}
In the case $K=2$  in \eqref{eq:cheegerconstant_norm} we have the celebrated Cheeger inequality (\cite{chung1996lectures})  relating the normalised Cheeger constant and the {second eigenvalue} of the normalised {supra-}Laplacian:
\begin{equation}
    \label{eq:CheegerK2}
     \norm{\cfunc h}_2 \leq \sqrt{2\norm{\eval}_{2,a}} \,.
\end{equation}
We state a result about the monotonicity of $\cfunc h_K$ in $K$, which is useful for making comparisons between differently sized packings/partitions produced by clustering algorithms. 
\begin{prop} \label{prop:h} For $K \in \mathbb{N}$ we have,
    \begin{equation}
    \begin{aligned}
        \norm {\cfunc h}_K &\leq \norm {\cfunc h}_{K+1}, \\
        {\cfunc h}_{K} &\leq {\cfunc h}_{K+1}. \\
    \end{aligned}
    \end{equation}    
\end{prop}
\begin{proof}
See Appendix~\ref{sec:prop_h_proof}.
\end{proof}

For general $K$, we recall the following result for the $K$-th smallest eigenvalue $\smash{ \norm{\eval}_{K,a} }$ of $\smash{ \nlap }$, which follows directly from  \cite[Thm.~4.9]{Lee2014}. For a graph $\sset G$ and given $\delta \in (0,1)$, there exists a $K'$-packing $\{\sset X_1, \dots ,\sset X_{K'}\}$, $\sset X_i \cap \sset X_j = \emptyset$, $\forall i \neq j$ for some $K' \geq \lceil(1-\delta)K\rceil$ 
such that $\smash{ \norm{\sfunc h}_{K'} \leq C K^{1/2} \delta^{-3/2}(\norm{\seval}_{K,a})^{1/2} }$ for some constant $C$ depending on the graph but independent of the packing. By choosing $\delta = 1/2K$ we have that $K' \ge K$ and using Proposition~\ref{prop:h} we obtain
    \begin{equation}
        \norm{\cfunc h}_{K} \leq \norm{\cfunc h}_{K'} \leq 2^{3/2}CK^2\sqrt{\norm{\eval}_{K,a}}\,. \label{eq:KwayCheeger}
    \end{equation}

For reasons to be discussed in Sec.~\ref{sec:eigenproblem} we will focus on the unnormalised {supra-}Laplacian $\slap$ in the following and in particular in the numerical investigations below. However, Cheeger inequalities involving the unnormalised Cheeger constant $\cfunc h_K$ and the unnormalised {supra-}Laplacian $\slap$ are less developed than their normalised counterparts, which we have just discussed. From \cite[Thm.~3.1]{Keller2016} we know that for a graph $\sset G = (\sset V, \sset E)$ equipped with adjacency $\sop W$ and an unnormalised Laplacian $\sop L$ with corresponding $k$-th eigenvalue $\seval_k$ we have the Cheeger inequality $\sfunc h_2 \le \sqrt{2\seval_2\max_x \sfunc d(x)}$ where $\sfunc d(x) = \sum_{y \in \sset V}\sop W_{x,y}$.
Thus, for $K=2$ we have the following:
\begin{equation}
    \label{eq:CheegerIneqUnnorm}
    \cfunc h_2 \le \sqrt{2\eval_{2,a} \max_{(t,x)} \adep{\func d}(t,x)}.
\end{equation}
To compare the relative tightness of \eqref{eq:CheegerIneqUnnorm} and \eqref{eq:CheegerK2}, let us briefly assume that the degrees are constant, say, $\adep{\func d}(x,t) \equiv d$. Then we see that $\cfunc h_K = d\, \norm{\cfunc h}_K$ and $\smash{ \slap = d\,\nlap }$, the latter implying~$\eval_{a,K} =  d\,\norm{\eval}_{a,K}$. 
The unnormalised Cheeger inequality~\eqref{eq:CheegerIneqUnnorm} in this case becomes $\cfunc h_2 \le \sqrt{2\eval_{2,a} d}$, which is---by the equations listed in the previous sentence---equivalent to the Cheeger inequality~\eqref{eq:CheegerK2} in the normalised case.
We thus expect the performance guarantee given by the unnormalised Cheeger inequality to be as strong as the one from the normalised case, as long as the relative variation in degrees is not too large.

\section{Eigenvalues of supra-Laplacians and the hyperdiffusion limit}\label{sec:eigenproblem}

The graph cuts defined in the previous section critically depend on the eigenvalues of the associated supra-Laplacians.
The eigenvectors of the (un)-normalised supra-Laplacian are used ahead to perform spectral partitioning of the spacetime graph to detect persistent communities. The supra-Laplacian depends on the parameter $a$, which scales the contribution of the spatial and temporal components. 
We show in Thm.~\ref{thm:var} that a few of the eigenvalues(vectors) of the corresponding unnormalised supra-Laplacian $\smash{\ulap}$ approach that of the corresponding unnormalised dynamic Laplacian $\smash{\unnorm \dlap}$ as~$a \to \infty$. For the normalised Laplacian $\smash{\nlap}$, the eigenvalues approach {those} of the normalised temporal Laplacian $\norm \ltemp$ as~$a \to \infty$, as shown in Thm.~\ref{thm:normL}.

\subsection{The unnormalised Laplacian $\mathbfcal L^{(a)}$}
Let $\sset G' = (\sset V', \sset E', \sop W')$ be such that $\sop W'$ is defined in \eqref{eq:Wtemp_op}. The vertices $\sset V'= \{1,\ldots, T\}$ and $t \sim s \textnormal{ in } \sset E' \iff \sop W'_{t,s} \neq 0$. Let $\sop L'$  be the corresponding graph Laplacian. We denote the $k$-th eigenpair of $\sop L'$ by $(\seval_k^{\rm temp},\sfunc f_k^{\rm temp})$. 
The $TN$ eigenpairs $\{ (\eval_{1,a}, \func f_{1,a} ), \ldots, (\eval_{TN,a}, \func f_{TN,a}) \}$  of $\ulap$ can be split into two collections: temporal $\{ (\eval^{\rm temp}_{2,a}, \func f^{\rm temp}_{2,a} ), \ldots, (\eval^{\rm temp}_{T,a}, \func f^{\rm temp}_{T,a}) \}$ and spatial $\smash{ \{ (\eval^{\rm spat}_{1,a}, \func f^{\rm spat}_{1,a}), \ldots, (\eval^{\rm spat}_{TN-T+1,a}, \func f^{\rm spat}_{TN-T+1,a})\} }$.

\begin{lemma}[\bf Spatial and temporal eigenpairs of $\ulap$]
\label{lem:spattemp}
    The eigenpairs of $\ulap$ are classified as follows: 
    \begin{enumerate}
        \item $T-1$ \emph{temporal eigenvalues} $\unnorm \eval_{k,a}^{\rm temp}$, which satisfy $\unnorm \eval_{k,a}^{\rm temp} = a^2\seval_k^{\rm temp}$, $2\le k\le T$.
        Corresponding to these temporal eigenvalues are $T-1$ \emph{temporal eigenfunctions} $\smash{ \unnorm {\func F}_{k,a}^{\rm temp} }$, which satisfy $\unnorm {\func F}_{k,a}^{\rm temp}(t,x)=\sfunc f_k^{\rm temp}(t)$, $2\le k\le T$.
        \item $TN-T+1$ \emph{spatial eigenvalues} $\unnorm \eval_{k,a}^{\rm spat}$.  Corresponding to these spatial eigenvalues are $TN-T+1$ \emph{spatial eigenfunctions} $\unnorm {\func F}_{k,a}^{\rm spat}$.
        For each $1\le k\le TN-N+1$, there is a constant $C_k$ such that  $\sum_{x\in V} \unnorm {\func F}_{k,a}^{\rm spat}(t,x)=C_k$ for all $1\le t\le T$.
    \end{enumerate}
\end{lemma}
\begin{proof}
See Appendix \ref{sec:lem:spattemp}.
\end{proof} 
By convention we consider $\eval_{1,a}$ to be spatial eigenvalue; this is consistent with Thm.~\ref{thm:var} and previous work~(\cite{FrKo23}).
We now show {in Thm.~\ref{thm:var}} that the first $N$ spatial eigenvalues and eigenvectors approach those of the dynamic Laplacian in~\eqref{eq:dLap} %(\cite{FroylandKwok2015}) 
in the hyperdiffusion limit, i.e.\ as $a \to \infty$. Part 3 of {Thm.~\ref{thm:var}  appeared} in \cite{sole2013spectral}; we prove the theorem using a variational approach, which {in comparison to \cite{sole2013spectral}} results in a stronger statement about the behaviour of $\eval_{k,a}$ and~$\func F_{k,a}$ and generalises to different forms of~$\sop W'$.
Note that the spectrum $\set \sigma_a$ of $\ulap$ is given by $\smash{ \set \sigma_a = \bigcup_k \unnorm \eval^{\rm spat}_{k,a} \cup \bigcup_k \unnorm \eval^{\rm temp}_{k,a} }$ for spatial eigenvalues $\smash{ \unnorm \eval^{\rm spat}_{k,a} }$ and temporal eigenvalues $\unnorm \eval^{\rm temp}_{k,a}$. Ordering the eigenvalues $\unnorm \eval_{k,a}$ of $\ulap$ as $\unnorm \eval_{1,a} \leq \unnorm \eval_{2,a} \leq \dots$, we obtain the trivial inequalities, 
\begin{equation}
    \unnorm \eval_{k,a} \leq \unnorm \eval_{k,a}^{\rm spat} {\mbox{ for $1\le k\le TN-N-1$, and }} \ \ \unnorm \eval_{k,a} \leq \unnorm \eval_{k,a}^{\rm temp} {\mbox{ for $2\le k\le T$}}.
    \label{eq:order}
\end{equation}

\begin{theorem}[Hyperdiffusion limit of $\ulap$]
\label{thm:var}
Let $\set{G}=(\set{V},\set{E},\sadj)$  be a spacetime graph and {let $\unnorm \eval_{k,a}, k=1,\ldots,TN,$ be the eigenvalues of the unnormalised supra-Laplacian $\ulap$}. Let $\unnorm \seval_k^D$, {$k=1,\ldots,N$,} denote the eigenvalues of the corresponding  dynamic Laplacian~$\unnorm \dlap$. Then the following statements hold:
\begin{enumerate}
    \item $\eval_{k,a} \le \seval_k^D$, {for $1\le k\le N$}.
    \item $\eval_{k,a}$ is nondecreasing for increasing $a$, {for $1\le k\le TN$}.
    \item $ \lim_{a \rightarrow \infty}\unnorm \eval_{k,a} = \unnorm \seval_{k}^D$, {for $1\le k\le N$}.
    \item The accumulation points of the sequence $(\unnorm {\func F}_{k,a}^{\rm spat})_{a > 0}$ for $a\to\infty$ are vectors constant in time, {for $1\le k\le NT-T-1$.}
\end{enumerate}
\end{theorem}
\begin{proof}
   See Appendix \ref{sec:them:var}. 
\end{proof}

If we consider the $k^{\rm th}$ eigenvalue of the unnormalised Laplacian as the quality measure for the best $k$-packing, 
we can then interpret the eigenvalue inequality $\eval_{k,a} \le \seval^D_k$ as ``{the} best spacetime {$k$-}packing is always better than {the} best static {$k$-}packing''. Also, the difference between {these} two {packings} is monotonically vanishing %closed
as $a\nearrow \infty$, i.e., as we increasingly force similarity between adjacent time slices (statements 2 and~3). Finally, statement 4 says that again in the limit $a\nearrow \infty$, packings extracted from the $\func F_{i,a}^{\rm spat}$ become increasingly fixed in time.

\subsection{The normalised Laplacian $\overline{\mathbfcal L^{(a)}}$}

Now we turn to the normalised supra-Laplacian $\nlap$  and its hyperdiffusion limit. We will show that in this case only temporal information prevails, essentially due to the normalisation by the degrees which are for $a\gg 1$ dominated by the temporal contributions. This {greatly} reduces the usefulness of $\nlap$ for spatiotemporal clustering.

Consider the spacetime graph  $\set{G}^{\rm temp} := (\set V, \set E, \adjtemp)$.  Let the corresponding normalised temporal Laplacian be denoted by $\smash{ \norm \ltemp }$ and the normalised Laplacian acting over the node set~$\sset V'$ by $\norm {\sop L'}$. For functions $\sfunc f(t)$ and $\func f(t,x)$ defined on $\sset V'$ and $\set V$, respectively, we have that
\begin{align}
    \norm{\sop L'}f(t) &= \sum_s \sop W'_{t,s} \left(  \frac{\sfunc f(t)}{\sfunc d^{\rm temp}(t)} - \frac{\sfunc f(s)}{\sfunc d^{\rm temp}(t)^{1/2}\sfunc d^{\rm temp}(s)^{1/2}}\right) , \label{eq:littlenormL}\\
    \norm{\op L^{\rm temp}}\func F(t,x) &= \sum_{s} \adjtemp_{(t,x),(s,x)} \left( \frac{\func F(t,x)}{\sfunc d^{\rm temp}(t)} - \frac{\func F(s,x)}{\sfunc d^{\rm temp}(t)^{\frac{1}{2}}\sfunc d^{\rm temp}(s)^{\frac{1}{2}}} \right)\label{eq:normLtemp}.
\end{align}

Let $(\sfunc \mu_k, \sfunc \phi_k)$, $k=1,\dots, T$ be the corresponding eigenpairs of $\norm {\sop L'}$. Let $\func F_k(t,x) := \sfunc \phi_k(t)$ (for all $x$) be the lift of the eigenfunction $\sfunc \phi_k$ in spacetime. Then,
    \begin{equation}
        \begin{aligned}
            \sfunc \mu_k \func F_k(t,x) &= \mu_k \sfunc \phi_k(t) = \norm {\sop L'} \sfunc \phi_k(t) \\
            &\stackrel{\eqref{eq:littlenormL}}{=} \sum_{s} \sop W'_{t,s} \left( \frac{\sfunc \phi_k(t)}{\sfunc d^{\rm temp} (t)} - \frac{\sfunc \phi_k(s)}{\sfunc d^{\rm temp}(t)^{\frac{1}{2}}\sfunc d^{\rm temp}(s)^{\frac{1}{2}}} \right) \\
            & \stackrel{\eqref{eq:Wtemp_op}}{=} \sum_{s} \adjtemp_{(t,x),(s,x)} \left( \frac{\func F_k(t,x)}{\sfunc d^{\rm temp}(t)} - \frac{\func F_k(s,x)}{\sfunc d^{\rm temp}(t)^{\frac{1}{2}}\sfunc d^{\rm temp}(s)^{\frac{1}{2}}} \right) \\
            & \stackrel{\eqref{eq:normLtemp}}{=} \norm {\ltemp} \func F_k(t,x).
        \end{aligned}
    \end{equation}
Thus $\func F_k(t,x) = \sfunc \phi_k(t)$ is an eigenvector of $\norm {\ltemp}$ and is independent of $x$. Further, the space $\mathbb{E}_k := \mathrm{span}\{\sfunc \phi_k\} \otimes \mathbb R^N = \{\func F \,:\, \func F(t,x) = \sfunc \phi_k(t) \sfunc v(x) \text{ with } \sfunc v{(x)} \in \mathbb{R}^N\}$ is the $N$-dimensional eigenspace of $\smash{\norm {\ltemp}}$ corresponding to eigenvalue $\sfunc \mu_k$.
We have the following theorem relating $\smash{ \norm{\ltemp} }$ and~$\nlap$.
\begin{theorem}[Hyperdiffusion limit of $\nlap$]
\label{thm:normL}
   Let $\set G = (\set V, \set E,  \sadj)$ be a spacetime graph and let  $(\norm{\eval}_{k,a},\norm{\func F}_{k,a})$  be an eigenpair of the normalised supra-Laplacian $\smash{\nlap}$, for some $k\in\{1,\ldots,TN\}$. Then in the limit $a \to \infty$ we have that every accumulation point of this eigenpair is of the form $(\sfunc \mu_j,\func F)$, where $\func F\in \mathbb{E}_j$, for some $j \in \{1,\dots,T\}$.
\end{theorem}
\begin{proof}
    See Appendix~\ref{sec:thm:normL}.
\end{proof}

As a consequence of Thm.~\ref{thm:normL}, the first $N$ eigenvalues of $\nlap$ are converging to % equal
to zero as $a \to \infty$. Moreover, there is no clustering-relevant structure in the corresponding eigenfunctions.
For $\slap$ we obtain more meaningful information as $a \to \infty$. From Thm.~\ref{thm:var} we know that as $a \to \infty$, the first $N$ eigenvalues of $\slap$ {recover those of} $\sop L^D$ which  has only one trivial zero eigenvalue if the graph described by $\sop W^D$ is connected. Moreover, the eigenvectors of $\slap$ also maintain structure as $a\to\infty$, namely informing us about optimal static clusterings.
{For these reasons} we use the unnormalised Laplacian $\slap$ for spectral partitioning of spacetime graphs. 
The use of the unnormalised Laplacian is also justified in \cite[Thm.~21]{vonLuxburg2008} if the eigenfunctions of $\slap$ used for spectral clustering have corresponding eigenvalues that lie outside the range $(\min_{t,x}\adep {\func d}(t,x), \max_{t,x}\adep {\func d}(t,x))$. This condition is satisfied for all {selected} eigenfunctions in all examples presented in Secs.~\ref{sec:example_networks} and~\ref{sec:sen}.

\section{Spectral partitioning in spacetime graphs using the supra-Laplacian}\label{sec:matrix_alg}

In the previous sections we discussed how the quality of graph cuts in spacetime graphs depends on the smallest nonzero eigenvalues of the associated supra-Laplacians.
In this section we describe how to perform spectral partitioning with sparse eigenbasis approximation (SEBA) by using the first $R \in \mathbb{N}$ spatial eigenvectors of the unnormalised supra-Laplacian $\slap$ to construct approximate minimisers of the Cheeger constant. In particular, we will discuss the choice of the temporal diffusion strength $a$.

We formulate the spectral partitioning Algorithm \ref{alg:spec_part} by defining a matrix representation of $\slap$ and in Secs.~\ref{sec:example_networks} and \ref{sec:sen} we demonstrate Algorithm \ref{alg:spec_part} on a few example networks and a network of time-varying US Senate roll call votes, which has been  analysed in~\cite{Mucha2010, Waugh2009}.

\subsection{Matrix representations}
\label{ssec:matrixform}

Consider a general graph $\sset G = (\sset V,\sset E, \adj)$. We denote by $\smat W$ the matrix form of the adjacency $\adj$ and by $\smat L$ the matrix form of the corresponding unnormalised Laplacian~$\lap$ with respect to the canonical basis. 
We further define the degree matrix $\smat D$ by 
\begin{equation}
    {\smat D}:= {\rm diag}\left( \textstyle{\sum_{y}} \smat W_{1,y},  \dots, \textstyle{\sum_{y}} \smat W_{N,y} \right).
    \label{eq:mat_degreevolume}
\end{equation}
Then $\smat L$ can be written as
\begin{equation}
    {\smat L} = {\smat D} - {\smat W}.
\end{equation}
We have the following equivalence with $\lap$: 
\begin{equation}
    \lap \sfunc f(i) = (\smat L \svect g)_i 
    \textnormal{ where } \svect g = [\sfunc f(1), \sfunc f(2), \dots ,\sfunc f(N)]^\top.
\end{equation}

\subsubsection{Spacetime matrices}

We now extend these ideas to matrix forms of spacetime graph operators. Recall that $\set G$ is a multiplex spacetime graph and that $N$ is the number of vertices at any time and $T$ is the number of time slices. We first define a {bijection $\func i: \mathbb{R}^T\times \mathbb{R}^N\to \mathbb{R}^{TN}$} 
\begin{equation}
    \label{eq:lexicographic_ordering}
    \func i: (t,x) \mapsto N(t-1) + x
\end{equation}
on the vertices $(t,x)$ which allows us to define spacetime vectors $\vect F \in \mathbb R^{TN}$. Thus, $\func i(\set{V}) = \{1,2, \dots, N, N+1, \dots, TN-1, TN\}$.
Next, adjacency and Laplacian matrices are constructed such that the rows and columns are ordered  with respect to~\eqref{eq:lexicographic_ordering}.
The matrix form $\sadjmat$ of $\sadj$ is then given by the $TN\times TN$ matrix
\begin{equation}
\label{eq:adj_mat_assembly}
    \sadjmat = \begin{bmatrix}
        \sadj_{(1,1),(1,1)} & \dots & \sadj_{(1,1),(T,N)} \\
        \vdots & & \vdots \\
        \sadj_{(T,N),(1,1)} & \dots & \sadj_{(T,N),(T,N)}
    \end{bmatrix}.
\end{equation}
From this definition, we can split $\sadjmat$ into spatial and temporal components, analogously to $\sadj$ in~\eqref{eq:supraW}. Thus,
\begin{equation}
    \sadjmat := \adjspatmat + a^2 \adjtempmat,
    \label{eq:adj_mat_split}
\end{equation}
where the spatial and temporal matrix forms  $\adjspatmat$ and $\adjtempmat$ are defined below. The matrix form of $\adjspat$ follows from the above representation as a block diagonal matrix,
\begin{equation}
    \adjspatmat   = \begin{bmatrix}
        \smat W^{\rm spat}_1 & \dots & 0 \\
        \vdots & \ddots & \vdots \\
        0 & \dots & \smat W^{\rm spat}_T
    \end{bmatrix},
    \label{eq:mat_Wspat_arbitrary}
\end{equation}
where each
\begin{equation}
    \smat W^{\rm spat}_t = \left[ \sadjmat_{(t,x),(t,y)} \right]_{x,y}, \quad  t=1,\ldots,T
\end{equation}
is an $N\times N$ matrix.
The block diagonal structure of $\mat W^{\rm spat}$  is shown in Fig.~\ref{fig:intro} (right). Let  $\smat W' \in \mathbb{R}^{T \times T}$ be a symmetric matrix representing 
the operator $\sop W'$ from~\eqref{eq:Wtemp_op} in the canonical basis. Again using the indexing from \eqref{eq:lexicographic_ordering}, the matrix form $\adjtempmat$ of $\adjtemp$ is
\begin{equation}
    \adjtempmat = \smat W' \otimes \smat I_N,
    \label{eq:mat_Wtemp}
\end{equation}
where $\smat I_{N}$ is the $N \times N$ identity matrix and  $\otimes$ is the standard Kronecker product given by
\begin{equation}
\label{eq:Kronecker}
   \mathbb{R}^{N_1N_2 \times M_1M_2}\text{\reflectbox{$\in$}} \smat A \otimes \smat B := \begin{bmatrix}
\smat A_{1,1}\smat B & \dots &  \smat A_{1,M_1} \smat B \\ 
\vdots & \ddots & \vdots \\ 
\smat A_{N_1,1}\smat B & \dots & \smat A_{N_1,M_1}\smat B  \\ 
\end{bmatrix}, \ \ \smat A \in \mathbb{R}^{N_1 \times M_1}, \smat B\in \mathbb{R}^{N_2\times M_2}.
\end{equation}
Recalling the definitions in \eqref{eq:deg} we define diagonal degree and volume matrices via
\begin{equation}
\label{eq:mat_deg_vol}
\begin{aligned}
    \adep{\mat{D}}_{\func i(t,x), \func i(t,x)} &:=  \adep{\func d}(t,x), \\
    \mat{D}^{\rm spat}_{\func i(t,x), \func i(t,x)} &:= \func d^{\rm spat}(t,x), \\ 
    \mat{D}^{\rm temp}_{\func i(t,x), \func i(t,x)} &:= \sfunc d^{\rm temp}(t). 
\end{aligned}
\end{equation}
Note that vertices $(t,x) \in \set V$ along the respective diagonals are ordered as per~\eqref{eq:lexicographic_ordering}.  The matrix form $\slapmat$ of the graph Laplacian $\slap$ is
\begin{equation}
\begin{aligned}
    \slapmat &= \adep{\mat{D}} - \sadjmat\label{eq:mat_supL} \\  
                  &= (\mat{D}^{\rm spat} - \adjspatmat) + a^2 (\mat{D}^{\rm temp} - \adjtempmat) \\ 
                  &=:  \lspatmat + a^2 \ltempmat.
\end{aligned}
\end{equation}
The matrix form $\ltempmat$ of $\ultemp$ from~\eqref{eq:lspat_ltemp} can also be written as
\begin{equation}
 \ltempmat = \smat L' \otimes \smat I_N,
 \label{eq:mat_ltemp}
\end{equation}
where  $\smat L'$ is the Laplacian created from $\smat W'$. For an eigenfunction $\func F_{k,a}$ of $\slap$, we define the vector form $\vect {F}_{k,a}$ by
\begin{equation}
    \vect {F}_{k,a} := [\func F_{k,a}(1,1), \dots, \func F_{k,a}(1,N), \func F_{k,a}(2,1), \dots,  \func F_{k,a}(T,1), \dots, \func F_{k,a}(T,N)]^\top,
    \label{eq:mat_evec}
\end{equation}
which is  the corresponding eigenvector of $\slapmat$. 

\subsection{A spectral partitioning algorithm}\label{sec:alg}

To identify clusters in a given spacetime multiplex graph, including their appearance and disappearance through time, we begin by constructing the supra-Laplacian matrix $\slapmat$ from a collection of adjacency matrices $\{\smat W^{\rm spat}_t\}_{t=1,\ldots, T}$.
The algorithm that proceeds from the construction of the supra-Laplacian to obtaining spacetime packings is summarised below.
The individual steps in the algorithm are elaborated in the paragraphs that follow.

\IncMargin{1em}
\begin{algorithm}[h]
\caption{Spectral partitioning using the supra-Laplacian.}\label{alg:spec_part}
\SetKwData{Left}{left}\SetKwData{This}{this}\SetKwData{Up}{up}
\SetKwFunction{Union}{Union}\SetKwFunction{FindCompress}{FindCompress}\SetKwFunction{SEBA}{SEBA}
\SetKwInOut{Input}{Input}\SetKwInOut{Output}{Output}
\Input{Number of vertices $N$, number of time slices $T$, graph $\set G$ represented by a collection of spatial weights 
$\{ \smat W^{\rm spat}_t\}_{t=1,\dots, T}$, $\smat W^{\rm spat}_t \in \mathbb{R}^{N \times N}$, and temporal weights $\smat W' \in \mathbb{R}^{T \times T}$.} 
\Output{ $K$-packing $\{\set{X}_{n_1},\ldots,\set{X}_{n_K}\}$.}
\BlankLine
Let $\sadjmat$ be constructed using $\{\smat W^{\rm spat}_t\}_{t=1, \dots, T}$ as per \eqref{eq:adj_mat_split}, \eqref{eq:mat_Wspat_arbitrary} and \eqref{eq:mat_Wtemp}\;
Construct $\adep{\mat D}$ and $\slapmat$ using \eqref{eq:mat_deg_vol} and~\eqref{eq:mat_supL}\;
Identify leading nontrivial spatial and temporal eigenvalues {using Lemma~\ref{lem:spattemp}} and compute $a =  a_c$ such that $\smash{\eval^{\rm spat}_{2,a_c} \lesssim \ \eval^{\rm temp}_{2,a_c}}$ (for instance via bisection)\;
Compute $\smash{ (\eval^{\rm spat}_{k,a}, \vect F^{\rm spat}_{k,a}) }$, the $k$-th spatial eigenpair of $\slapmat$, for $k=2,\ldots,R+1$. Select a suitable value for $R$ as described in step 4 in the main text\;
Construct companion eigenvectors $\tilde{\vect F}_{k,a}$: Define $\smash{ {\svect f}_{k,a} := [\|{\vect F}^{\rm spat}_{k,a}( 1, \cdot)\|, \ldots, \|\vect F^{\rm spat}_{k,a}(T, \cdot)\|]\in\mathbb{R}^{T} }$, $k=2,\dots, R+1$. Define the spacetime extension $\tilde{\vect F}_{k,a} := {\svect f}_{k,a} \otimes \mathbbm 1_N$\; 
Isolate spacetime packing elements: Apply the SEBA algorithm (\cite{Froyland2019}) to the collection $\{ \smash{ \vect F^{\rm spat}_{2,a}, \tilde{\vect F}_{2,a},  \dots,  \vect F^{\rm spat}_{R+1,a}, \tilde{\vect F}_{R+1,a} }\}$. The output SEBA vectors are denoted $\{\vect S^{R}_{j,a}\}_{j=1,\dots,2R}$\; 
Identify spurious SEBA vectors:  If the $j^{th}$ SEBA vector is approximately constant on each index block corresponding to a single time fibre, i.e.\ if $[\vect s^{R}_{j,a}]_z\approx C_{j,t}$ for $z=\func i(t,x)$, $1\le x\le N$, $1\le t\le T$\;
Define packing elements:
Let $\{{n_1}
,\ldots,{n_{K}}\}$, $K\le 2R$ denote indices of the non-spurious SEBA vectors. Define $\set{X}_{n_j}=\{(t,x): [\vect S^{R}_{n_j,a}]_z>0, z=\func i(t,x)\}$ for each $j=1,\ldots,K$.  If a single spacetime vertex $(t,x)$ belongs to two or more $\set{X}_{n_j}$, assign it only to the $\set{X}_{n_j}$ with the largest $[\vect s^{R}_{n_j,a}]_z$ value. One may augment the packing with the unclustered set $\smash{ \set \Omega = \{(t,x) \,:\, [\vect s^{R}_{n_j,a}]_z = 0, z = \func i(t,x)\}}$ to obtain a $K+1$ partition, where $\func i$ is the ordering from~\eqref{eq:lexicographic_ordering}.
\vspace{0.1cm}
\end{algorithm}\DecMargin{1em}

\paragraph{Identify spatial and temporal eigenvalues and choose the temporal diffusion strength parameter $a$ (step 3).} 
We compute eigenpairs of $\slapmat$ denoted by $(\eval_{k,a}, \vect F_{k,a})$ for $k=1, \dots, TN$. Spatial eigenvalues are identified by locating the corresponding eigenvectors that have constant spatial means as per Lemma~\ref{lem:spattemp}. That is, eigenvectors $\vect f_{k,a}$ such that $\sum_{x \in V}\vect f_{k,a}(t,x) = C_k$, for all $ 1 \le t \le T$. This is {verified numerically} as in \cite[Sec.~4.2]{FrKo23} by computing the variance of the means of the time fibres $\vect f_{k,a}(t,\cdot)$; the eigenvector is spatial if this value is close to zero.
We now fix a value of~$a$. We do so by finding the smallest value of $a$ for which $\smash{ \eval^{\rm spat}_{2,a} \lesssim \eval_{2,a}^{\rm temp} }$. In the case of the unnormalised Laplacian $\ulap$, this choice of $a$ is unique as the eigenvalues are monotonic in $a$, spatial eigenvalues $\smash{ \eval_{k,a}^{\rm spat} }$ saturate as $a \to \infty$ by Thm.~\ref{thm:var}, and the temporal eigenvalues scale with~$\mathcal O(a^2)$. A similar heuristic was used in the continuous-time supra-Laplacian in~\cite[Sec.~4.1]{FrKo23}. The non-multiplex case in Sec.~\ref{sec:nonmultiplex}  will be treated in a similar spirit.

\paragraph{Choosing the number of spatial eigenvectors $R$ to input to SEBA (step 4).} To spectrally partition using SEBA, we require the parameter $R$, which is the number of {nontrivial} spatial eigenvectors of the supra-Laplacian to be used. The first $R$ spatial eigenvectors and corresponding $R$ companion vectors (step 5, described below) {yield} $2R$ SEBA vectors. 
For a selected value of $R$ we then follow the remaining steps 5--8 of Algorithm \ref{alg:spec_part} to construct a $K$-packing of the network, where $K \leq 2R$. We determine a final value for $R$ using a balance of two criteria:

\begin{enumerate}
    \item {\bf Spectral gap}: After fixing $a$ as per step 3 of Algorithm~\ref{alg:spec_part}, we compute eigenvalues $\unnorm \eval_{k,a}^{\rm spat}$ and look for the first significant gap between these values. If the first significant gap arises between eigenvalues \smash{$\unnorm \eval_{k,a}^{\rm spat}$ and $\unnorm \eval_{k+1,a}^{\rm spat}$}, then $R$ is chosen to be~$k-1$. This is based on the spectral gap theorem, see \cite[Thm.~4.10]{Lee2014}.
    \item {\bf Mean Cheeger ratio $\tilde {\cfunc H}(R)$}:
    For each value of $R \in \mathbb N$, we obtain $2R$ SEBA vectors, each with support $\set X_k$.   We compute the mean Cheeger ratio $\tilde {\cfunc H}(R) = \frac{1}{2R}\sum_{k=1}^{2R}\cfunc H(\set{X}_k)$ where $\cfunc H(\set X_k)$ is defined as per \eqref{eq:Cheeger}. The value of $R$ is then chosen as $R = \arg \min_R \tilde {\cfunc H}(R)$.
\end{enumerate} 
In the examples presented in sections~\ref{sec:example_networks} and \ref{sec:sen}, we use each of these criteria to fix a choice of $R$. We prioritise the criterion which provides the most clear choice for $R$.

\paragraph{Constructing companion eigenvectors for input to SEBA (step 5).}
In the previous paragraph we discussed the choice of number of spatial eigenvectors $R$.
Because of the orthogonality relationships  between the different eigenvectors, it is possible that the leading nontrivial $R$ spatial eigenvectors encode more than $R$  spatiotemporal clusters.
To fully extract all relevant cluster information from $R$ eigenvectors, for each eigenvector $k=2,\ldots,R+1$ we create a companion vector $\tilde{\vect F}_{k,a}$ with $\tilde{\vect F}_{k,a}(t,x):= \|\vect F^{\rm spat}_{k,a}( t, \cdot)\|_2$, which is a spacetime vector that is constant on each time slice, with constant value given by the $\ell_2$ norm of the eigenvector on each time slice.
This procedure regenerates the additional degrees of freedom that are limited by orthogonality of the spatiotemporal eigenvectors, and is explained in greater depth in \cite[ Secs. 6.3.2 and 7.1.1]{AFK24} in the context of supra-Laplace operators on Riemannian manifolds with Neumann boundary conditions. 

\paragraph{Isolating packing elements using SEBA (step 6).} The Sparse Eigenbasis algorithm (SEBA) was introduced in \cite{Froyland2019} to disentangle cluster elements from eigenfunctions of Laplacians that represent packing/partitioning of a set. We summarise the algorithm here. Given a set of vectors $\{\svect v_1, \dots, \svect v_r\}$, $\svect v_i\in \mathbb{R}^m$ for some $m \in \mathbb N$, the algorithm computes a basis $\sset S$  (whose elements are sparse vectors) of a subspace that is approximately spanned by the $\svect v_i$, $i=1,\ldots,r$. This is done by computing the ideal rotation matrix $\smat Q$, which when applied to $\smat V:= [\svect v_1 | \svect v_2 | \cdots | \svect v_r] \in \mathbb{R}^{m\times r}$ produces a matrix $\smat S = [\svect s_1 | \svect s_2 | \cdots | \svect s_r] \in \mathbb{R}^{m\times r}$ that is sparse and has approximately orthogonal columns which define the SEBA vectors $\sset S = \{\svect s_1, \dots, \svect s_r\}$. More precisely, one solves the following nonconvex minimisation problem,
\begin{equation}
    \operatorname*{arg\,min}_{\smat S,\smat Q} \frac{1}{2}\|\smat V - \smat S \smat Q\|_F^2 + \mu \|\smat S\|_1,
\end{equation}
where $\|\cdot\|_F$ is the matrix Frobenius norm, $\|\cdot\|_1$ is the matrix 1-norm and $\mu$ is an appropriately chosen regularisation parameter  (the standard choice in \cite{Froyland2019} is $\mu=0.99/\sqrt{m}$). More details can be found in \cite{Froyland2019} and code is available at \url{github.com/gfroyland/SEBA}.
We apply the SEBA algorithm to the $R$ leading nontrivial spatial eigenvectors and their corresponding $R$ companion vectors $\{\vect F^{\rm spat}_{2,a}, \tilde{\vect F}_{2,a},  \dots,  \vect F^{\rm spat}_{R+1,a}, \tilde{\vect F}_{R+1,a}\}$ constructed in Step 5.
SEBA returns a collection of vectors $\smash{ \{\vect S_{j,a}^{R} \}_{j=1,\dots, 2R} }$ for a fixed choice of $R$ and $a$. The supports of these vectors are individual clusters that form a packing of  $\set V$.

\paragraph{Identifying spurious packing elements (step 7).}\label{par:Kpartition} 
We observe in our experiments that SEBA can produce spurious vectors $\unnorm {\vect S}_{j,a}^R$ that do not provide meaningful clustering information and correspond to undesirably high cut values $\cfunc H(\set X_k)$. This occurs when the true number of clusters $K$ lies somewhere strictly between $R$ and $2R$.  SEBA outputs $2R$ vectors in total, $K$ of which are useful vectors representing spacetime clusters, and $2R-K$ of which are spurious.
These spurious vectors are easily identified because they are constant on time fibres;  this is illustrated in \cite[Figure~25]{AFK24} in the manifold setting. In the graph setting, these vectors can be identified by being constant on time fibres or having a much larger cut value. The latter is illustrated in Sec.~\ref{sec:ex2}.  

\section{Networks with clusters that appear and disappear in time}
\label{sec:example_networks}

The spacetime graphs shown in this subsection are generated artificially, see Appendix~\ref{sec:graphgen} for details. We perform a spectral partitioning on the corresponding {supra-}Laplacian $\ulap$  using Algorithm~\ref{alg:spec_part}.
We intentionally chose small graphs having $N=20$ vertices and a smooth, non-abrupt transitioning of their cluster structure to provide non-trivial and challenging benchmark cases.
We compare our approach with one of the most established clustering strategies, the Leiden algorithm presented in \cite{Traag2019} (an improvement of the so-called Louvain algorithm in \cite{blondel2008fast}), {adapted} in two different ways to the evolving nature of the networks. 
{Using both visual comparison and Cheeger ratios as quality measures,} we show that both instances of the Leiden algorithm {do not} capture the spatiotemporal nature and evolution of the clusters.
{In contrast, our approach accurately detects all of the transitions that the clusters undergo, and identifies times when no clusters exist.}

\subsection{Example 1: {Disappearance of the unclustered set}}\label{sec:ex1}
The  graph  and results are shown in Fig.~\ref{fig:ex1}. 
{The unclustered set gradually disappears into two clusters (alternatively two clusters appear)} as the graph changes from a regular graph. Using our spacetime spectral partitioning technique, we expect to find {these two clusters} when present, and otherwise {find} the network to {be} unclustered.

\subsubsection{Graph construction}
 We construct a graph with $N=20$ and $T=21$.
  These graphs are represented by a sequence of adjacency matrices $\{\smat W^{\rm spat}_t\}$, $t=1,\dots,21$ such that $\smat [W^{\rm spat}_{t}]_{ij} \in \{0,1\}$, $\forall i,j,t$.
  At $t=1$, the graph is $11$-regular and slowly transitions to a graph at $t=21$ that contains two clusters that are connected to each other by a single edge.
  At some time around $t=13$, the unclustered set $\set \Omega = \{(t,x) \,:\,  1\leq t\leq 13, x = 1, \dots, 20\}$ disappears into $\{\set Y_1, \set Y_2\} $ where $\set Y_1 = \{ (t,x) \,:\, 14\le t \le 21, x= 1, \ldots, 5, 18,19,20 \}$ and  $\set Y_2 = \{ (t,x) \,:\, 14\le t\le 21, x= 6, \ldots, 17 \}$. 

\subsubsection{Cluster detection} We use Algorithm~\ref{alg:spec_part} to compute the individual clusters. 
Using Step 3 of Algorithm~\ref{alg:spec_part} we obtain the critical value $a=8.15921$ which is fixed for the rest of the computations.
The parameter $R$ tells one how many of the first $R$ nontrivial spatial eigenvectors $\vect F^{\rm spat}_{k,a}$, $k=2,\ldots,R+1$ should be used for spectral partitioning.  
At $R=1$ there is a significant spectral gap between $\unnorm \eval^{\rm spat}_{2,a}$ and $\unnorm \eval^{\rm spat}_{3,a}$, see Fig.~\ref{fig:ex1} (lower right). 
\begin{figure}[htbp]
    \centering
    \includegraphics[width=\textwidth]{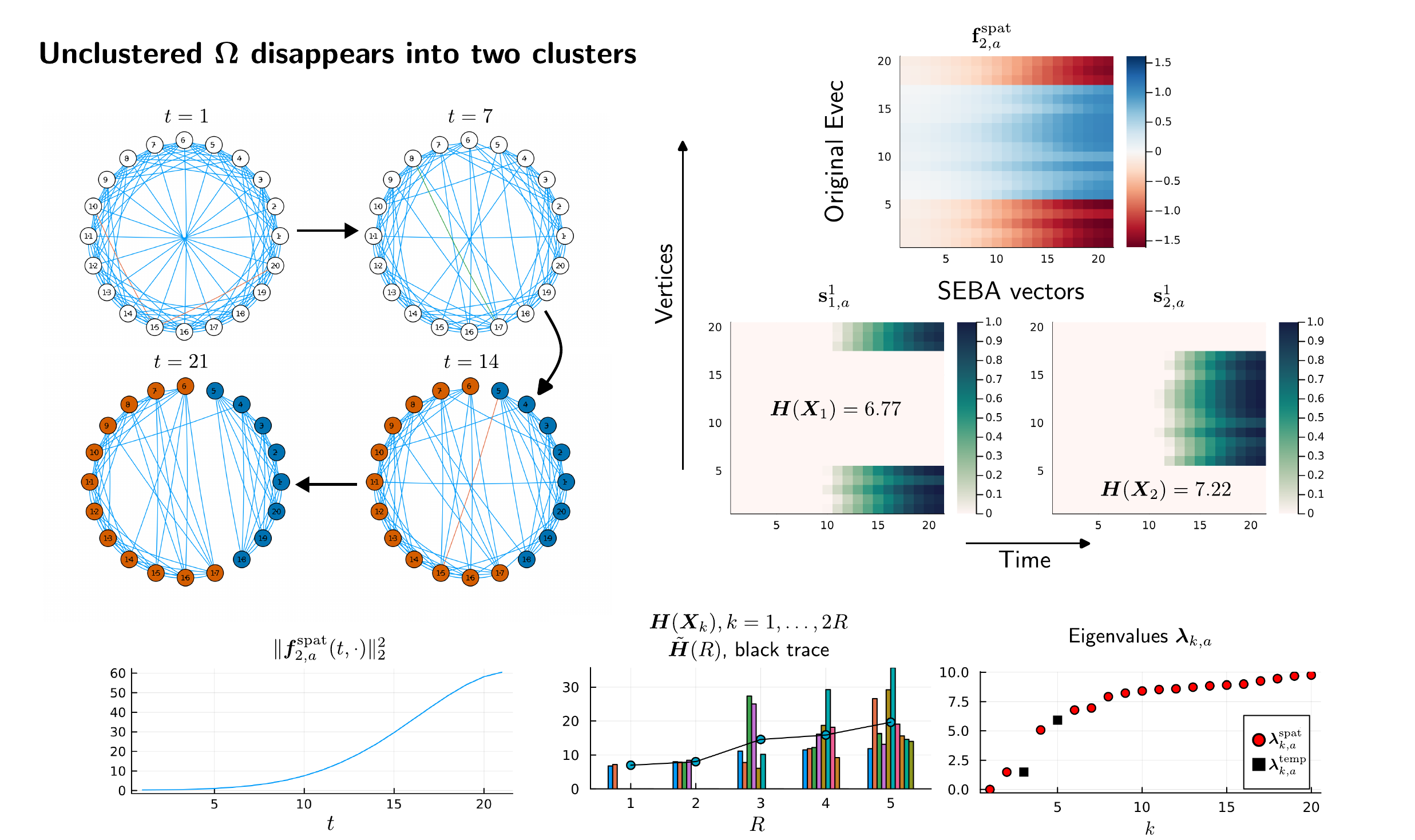}
    \caption{Spectral partitioning of spacetime graphs with the supra-Laplacian to discover the {disappearance of the unclustered set (and the appearance of two clusters)} discussed in in Sec.~\ref{sec:ex1}. (Upper left) Multiple time slices of the spacetime graph. The vertices are coloured according to assignments from individual SEBA vectors computed using the unnormalised Laplacian $\slapmat$. Red edges indicate edges at time $t$ that are removed at time $t+1$. (Upper right) Cluster detection using spectral partitioning. Using $R=1$, the spatial eigenvector $ \vect F^{\rm spat}_{2,a}$ of the supra-Laplacian $\smash{ \slapmat }$ and the corresponding two SEBA vectors $\smash{ \vect S^{1}_{1,a},  \vect S^{1}_{2,a}}$, supporting clusters $\set{X}_1, \set{X}_2$ respectively, with Cheeger {ratios} computed using~\eqref{eq:Cheeger}.  The complete lack of support of $\vect S^{1}_{1,a},  \vect S^{1}_{2,a}$ during times earlier than $t=10$ indicates the unclustered spacetime nodes $\set\Omega$. (Lower left) Plot of $\smash{ \|\vect F_{2,a}^{\rm spat}(t, \cdot)\|_2^2 }$ versus $t$, which demonstrates growing confidence in partitioning as time increases.  Also shown are  (Lower middle) plots of the cut function $\cfunc H(\set{X}_k)$ for $k=1 \dots 2R$, vs $R$, where $\set{X}_k$ is the support of SEBA vector $\vect S_{k,a}^{R}$ and the  (Lower right) eigenvalues $\eval_{k,a}$, sorted by spatial (red circles) and temporal (black squares). }
    \label{fig:ex1}
\end{figure}
The mean Cheeger ratios $\smash{\tilde {\cfunc H}(R)}$ in Fig.~\ref{fig:ex1} (lower center)  remain low for $R=1$ and~$R=2$.
We therefore fix $R=1$ because this is a good choice in terms of both spectral gap and low cut values.  
Algorithm~\ref{alg:spec_part} returns the first nontrivial spatial eigenvector $\unnorm {\vect F}^{\rm spat}_{2,a}$ of $\ulapmat$, the corresponding SEBA vectors $\smash{ \unnorm {\vect S}^{1}_{1,a}, \unnorm {\vect S}^{1}_{2,a} }$, and their associated associated clusters $\set{X}_1:= {\rm supp}(\unnorm {\vect S}^{1}_{1,a}),\set{X}_2:= {\rm supp}(\unnorm {\vect S}^{1}_{2,a}) $ and unclustered set $\set{\Omega} = \set V \setminus \{\set X_1, \set X_2\}$. 

For comparison, we also perform  cluster detection using the {\it Leiden} algorithm~(\cite{Traag2019}). The algorithm is an iterative aggregation-based cluster-identification procedure that maximises modularity. We stop the algorithm when there is no further increase in modularity. The algorithm is applied to our network in two ways: slice-by-slice and on the whole graph. In the slice-by-slice case, Leiden (the shorthand name used for the Leiden algorithm) is used to discover clusters on every time slice. We then use a matching technique to connect the various spatial clusters into spacetime clusters, see Appendix~\ref{sec:app:matching}. \textit{Ad hoc} matching techniques have been used before to link individual spatial clusters in time, see \cite{macmillan2020} for example. Our approach solves a maximum weight edge cover problem to maximise intercluster weights and optimise temporal matching. The results are shown in Fig.~\ref{fig:leiden-seba} (upper row).
We also report the values of the cut function $\cfunc H(\set X_k)$ for each of the clusters obtained from our spectral partitioning algorithm and Leiden. {To illustrate the robustness of Algorithm \ref{alg:spec_part} with respect to the choice of the parameter $a$ in Appendix~\ref{sec:app:varyinga} we summarise cluster assignments obtained using a very wide range of $a$ values.}

\subsubsection{Discussion of clustering}
In Fig.~\ref{fig:ex1} {(upper right)} we observe that the first nontrivial spatial eigenvector $\unnorm {\vect F}^{\rm spat}_{2,a}$ gives a clear spacetime partition that separates the two sets of spatial vertices $\{1, \ldots, 5, 18, \ldots, 20 \}$ (red, negative values) and $\{6, \ldots, 17\}$ (blue, positive values)  through time. {The increasing strength of this separation} can be visualised by the increase in $\|\unnorm {\vect F}^{\rm spat}_{2,a}(s, \cdot)\|$ as $s$ increases (see Fig.~\ref{fig:ex1}, lower left).  The SEBA vectors clearly separate the appearing clusters after around $t=10$ and have zero support in the early time interval $t=1$ to $t=9$ where the graphs on individual time slices are close to regular. We observe that the sets $\set X_1$ and $\set X_2$ recover the ground truth clusters $\set Y_1$ and $\set Y_2$ respectively.

{In Fig.~\ref{fig:leiden-seba} (upper panel)} we compare these results with those obtained using the Leiden algorithm {followed by our optimal matching algorithm from Appendix \ref{sec:matching}}.
The slice-by-slice Leiden computations show that the network is partitioned into three clusters while the network is still near-regular ($t=1$ to $t=7$). However at later times, the correct two clusters are recovered. Nevertheless, the value $\max_{1 \leq k \leq 3}\cfunc H(\set X_k) = \cfunc H(\set X_3) = 29.7$ for the Leiden slice-by-slice case is very large compared to that of the spectral partitioning method, where $\max_{1 \leq k \leq 3} \cfunc H(\set X_k) = \cfunc H(\set X_2) = 7.22$ (with the identification $\set{X}_3 = \set \Omega$).
Thus our approach provides a {superior} balanced cut {over} the Leiden algorithm. 

\begin{figure}[htbp]
    \centering
    \includegraphics[width=\textwidth]{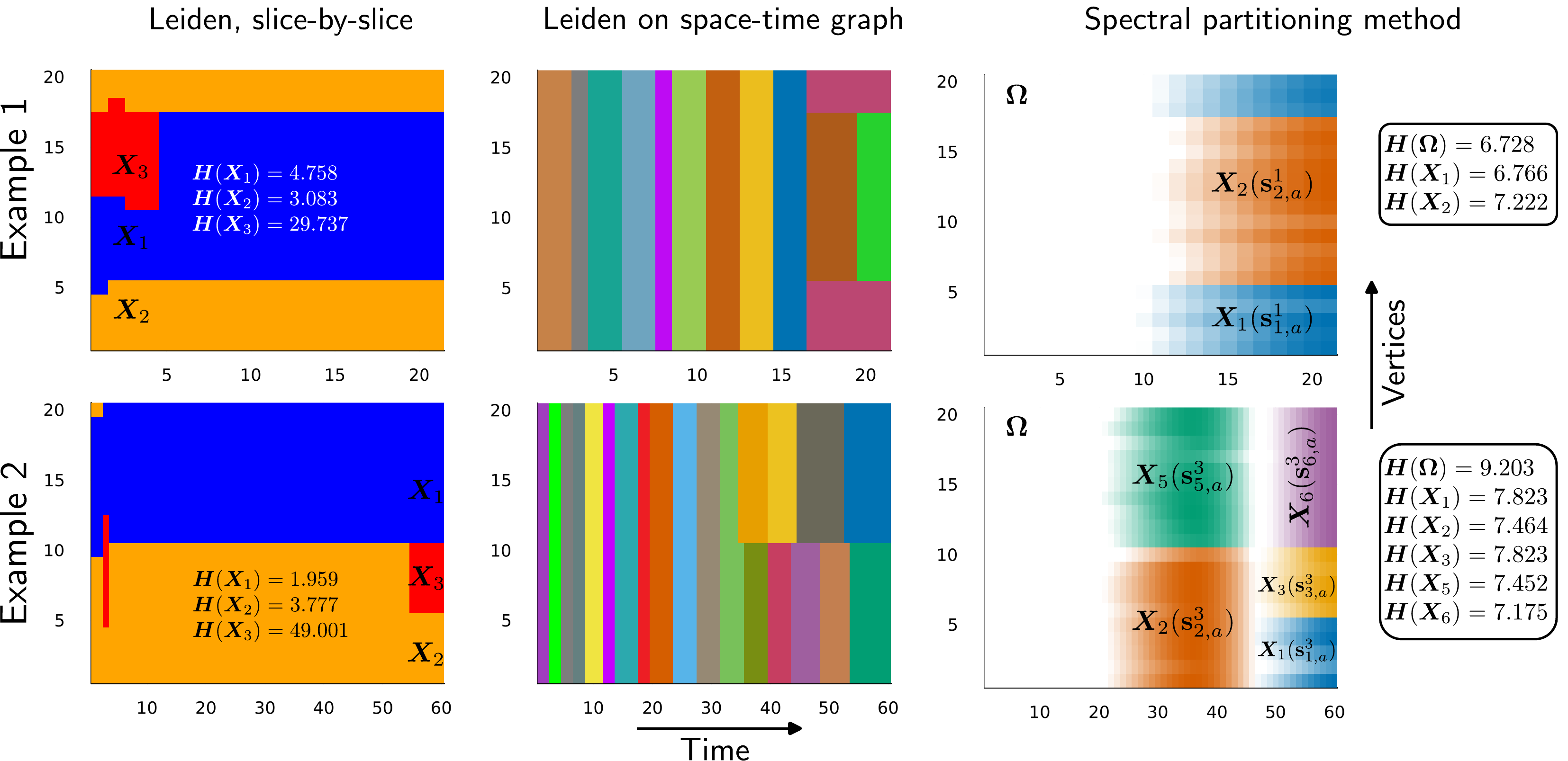}
    \caption{Comparing spectral partitioning with modularity maximisation using the Leiden algorithm (\cite{Traag2019}) for the two example networks. Leiden is performed for two scenarios of the spacetime graph: slice-by-slice (left column) and on the whole spacetime graph (middle column). Cheeger {ratios} are computed using \eqref{eq:Cheeger}, and shown for the slice-by-slice case only. Also shown (right column) are the individual clusters obtained from the spectral partitioning Algorithm~\ref{alg:spec_part} presented in Fig.~\ref{fig:ex1} and Fig.~\ref{fig:ex2}. In the first example (upper right), SEBA vectors $\vect S^{1}_{j,a}$ for $j=1,2$ and the non-empty unclustered set $\Omega$ are used to create a 3-partition of the network. In the second example (lower right), SEBA vectors $\vect S^{3}_{j,a}$ for $j=1,2,3,5,6$ and the non-empty unclustered set $\set \Omega$ are used to create a 6-partition of the network. }
    \label{fig:leiden-seba}
\end{figure}

Using Leiden on the spacetime graph (Fig.~\ref{fig:leiden-seba}, upper middle) yields a spatial separation of clusters only at the very late times $t=17$ to $t=20$, where the appropriate two clusters $\{1, \ldots, 5, 18, \ldots, 20 \}$ and $\{6, \ldots, 17\}$ are recovered. At times prior to $t=17$, spurious partitioning of many time slices into its own cluster renders these results unusable. 
The respective Cheeger ratio values $\cfunc H(\set X_k)$ of such clusters are $2a^2\approx 133.15$, which is significantly larger than the values for clusters obtained using our spacetime partitioning method.

\subsection{Example 2: {Disappearance of the unclustered set followed by another cluster disappearance}}\label{sec:ex2}

We extend the previous example by allowing one of the appearing clusters to {disappear into two new clusters}. 
We expect to resolve the appearing clusters using the eigenvectors of the corresponding supra-Laplacian. The results from the following example are presented in Fig.~\ref{fig:ex2}.

 \begin{figure}[htbp]
    \centering
    \includegraphics[width=\textwidth]{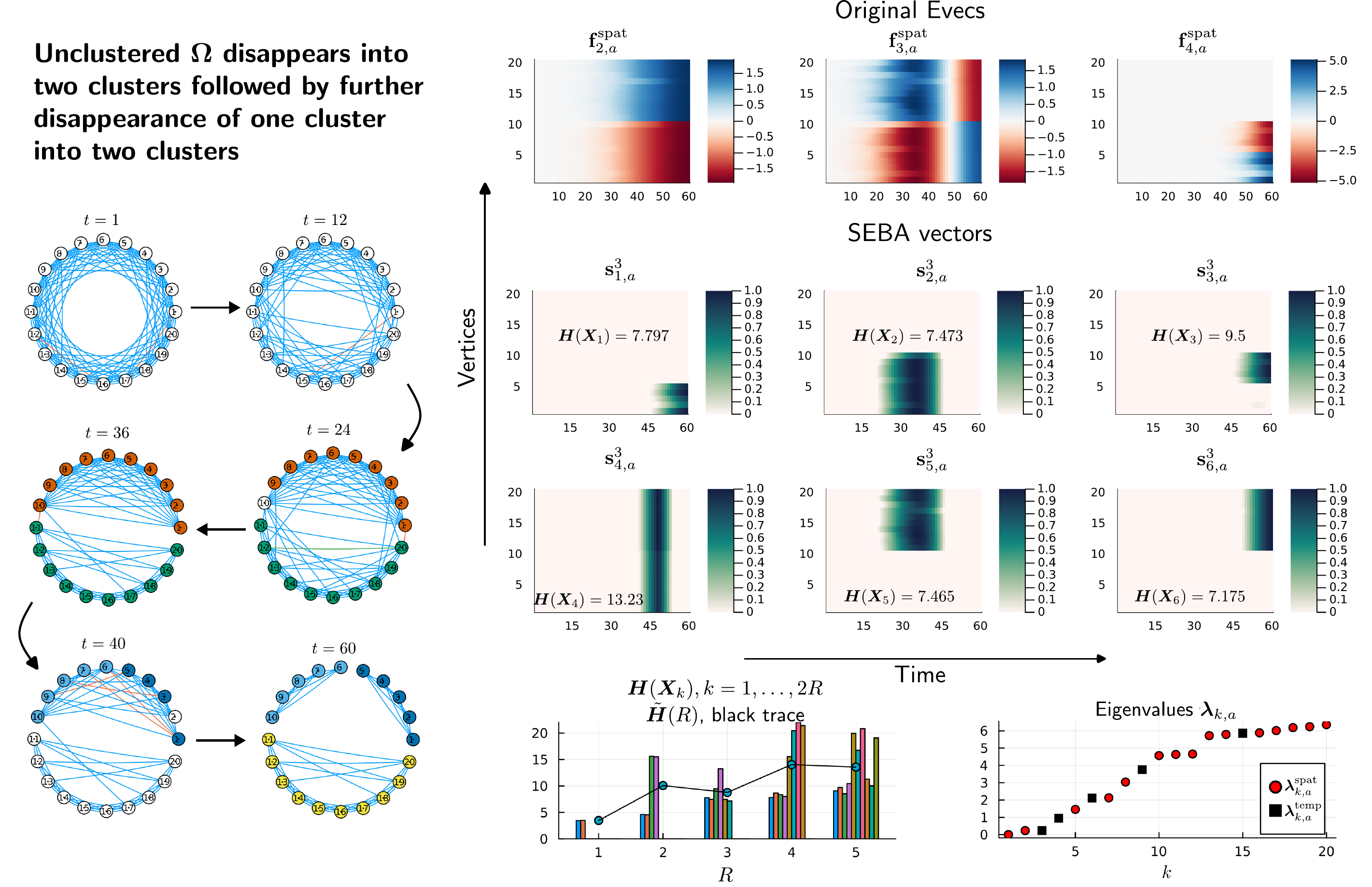}
    \caption{ Spectral partitioning of spacetime graphs with the supra-Laplacian $\slapmat$ to discover two successive cluster disappearance events. 
    (Left) Multiple time slices of the spacetime graph. The vertices are coloured according to assignments from individual SEBA vectors computed. Red edges indicate edges at time $t$ that are removed at time $t+1$. (Right) Cluster detection using spectral partitioning. (Upper right) First three spatial eigenvectors $\smash{\unnorm {\vect F}^{\rm spat}_{(2,3,4),a}}$ of the supra-Laplacian $\slapmat$. (Center right) Six SEBA vectors $\unnorm {\vect S}^{3}_{i,a}$ for $i=1,\ldots ,6$ computed using the three spatial eigenvectors above  with Cheeger {ratios} computed using \eqref{eq:Cheeger}.  Also shown (Lower right) are plots of the Cheeger {ratios} $\cfunc H(\set{X}_k)$ for $k=1, \dots, 6$, where $\set{X}_k$ is the support of SEBA vector $\unnorm {\vect S}_{k,a}^{3}$ and the eigenvalues $\unnorm \eval_{k,a}$, sorted by spatial (red circles) and temporal (black squares).}
    \label{fig:ex2}
\end{figure}

\subsubsection{Graph construction}
The graph is constructed with $N=20$ and $T=60$.
The graphs are represented by adjacency matrices $\{\smat W^{\rm spat}_t\}$, $t=1,\dots,60$, such that $\smat [W^{\rm spat}_t]_{ij} \in \{0,1\}$, $\forall i,j,t$.
At some time around $t=36$, the unclustered set $\set \Omega = \{(t,x) \,:\,  1\leq t\leq 36, x = 1, \dots, 20\}$ disappears into $\{\set Y_1, \set Y_2\} $ where $\set Y_1 = \{ (t,x) \,:\, 37\le t \leq 50, x= 1, \dots, 10 \}$ and  $\set Y_2 = \{ (t,x) \,:\, 37\le t\le 50, x= 11, \dots, 20 \}$. Next, around $t=50$, $\set Y_1$ disappears into $\{\set Y_3, \set Y_4\}$ where $\set Y_3 = \{(t,x) \,:\, t>50 , x = 1, \dots, 5\}$ and $\set Y_4 = \{(t,x) \,:\, t>50 , x = 6, \dots, 10\}$  
A few snapshots of the network are shown in Fig.~\ref{fig:ex2} (left panel).

\subsubsection{Cluster detection}  
We use Algorithm~\ref{alg:spec_part} to compute the individual clusters. 
Using Step 3 of Algorithm~\ref{alg:spec_part} we obtain the critical value $a=9.27364$ which is fixed for the rest of the computations.
The parameter $R$ tells one how many of the first $R$ nontrivial spatial eigenvectors $\vect F^{\rm spat}_{k,a}$, $k=2,\ldots,R+1$ should be used for spectral partitioning.  
At $R=1$ there is a significant spectral gap between $\unnorm \eval^{\rm spat}_{2,a}$ and $\unnorm \eval^{\rm spat}_{3,a}$. 
There is also a large gap  at $R=3$, between $\unnorm \eval^{\rm spat}_{4,a}$ and $\unnorm \eval^{\rm spat}_{5,a}$, see Fig.~\ref{fig:ex2} (lower right). 
The mean Cheeger ratios $\tilde {\cfunc H}(R)$ plotted in Fig.~\ref{fig:ex2} (lower right, black trace) shows a global minimum at $R=1$ and a local minimum at $R=3$. Therefore, reasonable choices of $R$ would be $R=1$ and $R=3$. 
The former choice would % reveal the 2-splitting 
only reveal the disappearance of the unclustered set, which
we have already discussed in Example 1.
We therefore set $R=3$, which provides further information, and will identify additional \mkadd{cluster transitions}.
Algorithm~\ref{alg:spec_part} returns the first three nontrivial spatial eigenvectors $\unnorm {\vect F}^{\rm spat}_{2,a}$ $\unnorm {\vect F}^{\rm spat}_{3,a}$, and $\unnorm {\vect F}^{\rm spat}_{4,a}$ of $\ulapmat$, the corresponding SEBA vectors $\smash{ \unnorm {\vect S}^{3}_{j,a}}, j=1,\ldots,6$, and the associated clusters $\set{X}_j := {\rm supp}(\unnorm {\vect S}^{3}_{j,a})$, $j=1,\ldots,6$ and unclustered set~$\set{\Omega}$. 

Similar to the previous example we also perform computations using the Leiden algorithm. We compare the results of the spectral partitioning method with those obtained using Leiden in Fig.~\ref{fig:leiden-seba} (lower). {As in the previous example, we illustrate the robustness of Algorithm \ref{alg:spec_part} with respect to the choice of the parameter $a$ in Appendix~\ref{sec:app:varyinga}, by summarising cluster assignments obtained using a very wide range of $a$ values.}

\subsubsection{Discussion of clustering}

Fig.~\ref{fig:ex2} (right panel) displays the vector $\unnorm {\vect F}^{\rm spat}_{2,a}$, which captures the slow disappearance of the unclustered set $\Omega$ into two clusters at approximately $t=36$.
These two clusters are represented by vertices $\{1, \dots, 10\}$ and $\{11,\ldots,20\}$ respectively. 
The vector $\unnorm {\vect{F}}^{\rm spat}_{3,a}$  captures a secondary cluster transition at approximately $t=48$, indicated by the sharp red/blue switch. 
The vector $\unnorm {\vect F}^{\rm spat}_{4,a}$ resolves a further breakup of the cluster $\{1,\ldots,10\}$ into the clusters $\{1,\ldots , 5\}$ and $\{6, \ldots , 10\}$ at around $t=48$.
The corresponding SEBA vectors $\unnorm {\vect S}^{3}_{j,a}$ successfully capture the $5$ individual clusters except for $\unnorm  {\vect S}^{3}_{4,a}$, which is a spacetime cluster that contains all vertices for the time that it exists. The SEBA vector $\unnorm  {\vect S}^{3}_{4,a}$ is identified as spurious by step 7 of Algorithm \ref{alg:spec_part}
and also corresponds to a high Cheeger {ratio} $\cfunc H(\set X_4)=13.23$ compared to the other clusters.
We therefore discard it to obtain a $K$-partition with $K=6$ (5 SEBA vectors and the unclustered set $\set \Omega$). 
By post-processing the cluster assignment with our optimal matching algorithm described in Appendix \ref{sec:matching}, one could relabel $\set{X}_6$ as $\set{X}_5$ and $\set{X}_3$ as $\set{X}_2$, to reduce $K$ from 6 to 3 and further reduce the Cheeger {ratios}. We recover $\set Y_1, \set Y_3$ and $\set Y_4$ from the sets $\set X_2$, $\set X_1$ and $\set X_3$ respectively. The set $\set Y_2$ is recovered from the union $\set X_5 \cup \set X_6$.

In Fig.~\ref{fig:leiden-seba} (lower panel) we show the slice-by-slice and full spacetime graph computations using the Leiden algorithm {followed by our optimal matching algorithm in Appendix \ref{sec:matching}}. The slice-by-slice clustering is not completely satisfactory in how it captures the emergence of two clusters followed by the splitting. The value $\max_{1\leq k \leq 3} \cfunc H(\set X_k) = \cfunc H(\set X_3) = 49.001$ is large compared to the clusters detected using the spectral partitioning algorithm ($\max_{1\leq k \leq 6} \cfunc H(\set X_k) = \cfunc H(\set \Omega) = 9.203$, with the identification $\set{X}_4 = \set \Omega$). Note that from Proposition~\ref{prop:h}, given a good candidate partition $\partn X = \{\set X_1, \ldots, \set X_K\}$ we should get higher values of $\max_k \cfunc H(\set X_k)$ for a 6-partition compared to a 3-partition. {In fact}, we see that the spectral partitioning algorithm outperforms Leiden, despite the former providing a $6$-partition and the latter providing only a $3$-partition.

\section{Networks with vertices that appear and disappear in time}\label{sec:sen}
In this section we use spectral partitioning on networks where \emph{vertices} appear and disappear as time passes. This requires 
using a non-multiplex form of a spacetime graph. 
We begin with 
%describing this framework 
a demonstration of our spatiotemporal spectral clustering using a toy model that contains vertices that appear and disappear. 
We then consider a real-world network of voting similarities between senators of the US Senate between the years 1987 and {2025}. {A subset of} this network was considered in \cite{Waugh2009, Mucha2010} to analyse polarisation in US politics. 
Our spatiotemporal spectral partitioning technique %in conjunction with our new construction 
reveals the \emph{onset} and \emph{magnitude} of polarisation as time passes, which to our knowledge has not been shown before.
%as per our knowledge.

\subsection{A non-multiplex framework}
\label{sec:nonmultiplex}

To work with vertices that appear and disappear in time, we propose a non-multiplex framework. Consider a network with $T$ time slices with each slice having a maximum of $N$ vertices that may appear or disappear in time. The network configuration is shown in Fig.~\ref{fig:toy_sen} (left).  Vertices are split into two types: present (red/blue) and absent (white). The entire collection of present spacetime vertices is denoted by $\set{V}$ with cardinality $N' >0$.
We create a superset of spatial vertices $$\hat {\sset V}:=\bigcup_{t=1}^T\{x: (t,x)\in\set{V}\},$$
and identify the vertices in $\hat {\sset V}$ with integers $1,\ldots,N$.
For each $t$, denote the vertices on the $t$-th time fibre by $\sset{V}_t:=\{x: (t,x)\in \set{V}\}\subset \{1,\ldots,N\}$, each with cardinality $N_t \leq N$. Thus, $\sum_{t=1}^T N_t = N'$.
For each $t=1,\dots,T$, the spatial network on $\sset{V}_t$ is given by a spatial adjacency  $ \smat W^{\rm spat}_t \in \mathbb R^{N_t \times N_t}$. Note that here $N' < TN$ while in the multiplex case $N'=TN$.
The non-multiplex 
%supra-adjacency
spacetime weights $\sadj$ on $\set V$ are then defined as
\begin{equation}
    \label{eq:adj_nonmultiplex}
    \begin{aligned}
        \sadj_{(t,x),(t,y)} = \smat [W^{\rm spat}_t]_{xy} &\textnormal{ if } x,y \in \sset V_t \, , \\
        \sadj_{(t,x),(t+1,x)} = a^2 &\textnormal{ if } (t,x), (t+1,x) \in \set V.
    \end{aligned}
\end{equation}
The resulting spacetime graph is no longer multiplex as not all vertices at a time slice are connected to their neighbours in the adjacent time slices. This is illustrated in Fig.~\ref{fig:toy_sen} (left). 

The matrix forms are defined as follows. First we define the lexicographic ordering $i$ such that
\begin{equation}
\begin{aligned}
    \func i &:(t,x) \mapsto \sum_{s=1}^{t-1}N_s+x, \\
    \func i &(\set V) = \left\{1, \dots, N_1, N_1+1, \dots, \textstyle \sum_{s=1}^{t}N_s-1, \sum_{s=1}^t N_s\right\}.
\end{aligned}
    \label{eq:lexicographic_nonmultiplex}
\end{equation}
The matrix $\sadjmat$ is assembled using the above ordering in the rows and columns, similar to~\eqref{eq:adj_mat_assembly}:
\begin{equation}
    \sadjmat = \begin{bmatrix}
        \sadj_{(1,1),(1,1)} & \dots & \sadj_{(1,1),(T,N_t)} \\
        \vdots & & \vdots \\
        \sadj_{(T,N_t),(1,1)} & \dots & \sadj_{(T,N_t),(T,N_t)}
    \end{bmatrix}.
    \label{eq:adj_mat_assembly_nonmultiplex}
\end{equation}
The supra-Laplacian matrix $\slapmat$ is constructed using~\eqref{eq:mat_supL}. We do not split the matrices $\slapmat$ and $\sadjmat$ into its temporal and spatial components, as the simple occupation structure from Fig.~\ref{fig:intro} is lost in the non-multiplex setting.
 
% where $\mat W^{\rm spat}_t \in \mathbb{R}^{N_t \times N_t}$ such that the various $N_t$ are not necessarily equal to each other.

\paragraph{Choosing $a$ for the non-multiplex framework.} 
As we do not have a clear distinction between spatial and temporal eigenvectors of the corresponding supra-Laplacian $\slapmat$, we cannot fix $a$ as in step 3 of Algorithm~\ref{alg:spec_part}. We use a new heuristic to determine $a$. In Algorithm~\ref{alg:spec_part} we balance spatial and temporal diffusion by finding $a$ such that $\unnorm \eval_{2,a}^{\rm spat} = \unnorm \eval_{2,a}^{\rm temp}$. We observe numerically that in the multiplex case, this heuristic yields an $a$ that is identical to the $a$ obtained by
%is equivalent to
balancing the spatial and temporal contributions of the Rayleigh quotient corresponding to the spatial eigenvalue. In other words we conjecture that 
\begin{equation}
    \label{eq:RayleighBalance}
    \eval_{2,a}^{\rm spat} = \eval_{2,a}^{\rm temp} \iff \langle \lspatmat \vect F^{\rm spat}_{2,a}, \vect F^{\rm spat}_{2,a} \rangle = a^2\langle \ltempmat \vect F^{\rm spat}_{2,a}, \vect F^{\rm spat}_{2,a} \rangle.
\end{equation}

In the multiplex case, there is no notion of $\smash{ \eval_{2,a}^{\rm temp} }$, and we cannot use the left-hand equality in (\ref{eq:RayleighBalance}) to select $a$.
Instead, 
we compute the value of $a$ for which $\langle \lspatmat \vect F_{2,a}, \vect F_{2,a} \rangle = a^2 \langle \ltempmat \vect F_{2,a}, \vect F_{2,a} \rangle$, where $\vect F_{2,a}$ is the second eigenvector of $\slapmat$.
\IncMargin{1em}
\begin{algorithm}[h!]
\caption{Algorithm to find spatial eigenvectors for the non-multiplex graph Laplacians.}\label{alg:spec_part_findspat}
\SetKwData{Left}{left}\SetKwData{This}{this}\SetKwData{Up}{up}
\SetKwFunction{Union}{Union}\SetKwFunction{FindCompress}{FindCompress}\SetKwFunction{SEBA}{SEBA}
\SetKwInOut{Input}{Input}\SetKwInOut{Output}{Output}
\Input{Diffusion constant $a$, non-multiplex supra-Laplacian $\slapmat \in \mathbb R^{N' \times N'}$. } 
\Output{Eigenvectors $\{\vect F_{k,a}\}_{k=k_1, k_2,\ldots, k_R}$ representing ``spatial'' eigenvectors of the non-multiplex network.}
\BlankLine
Compute eigenvectors $\vect F_{k,a} \in \mathbb R^{N'}$, $k=1,\ldots, N'$. Set $\smash{ \tilde {\svect f}_j^{\rm temp} }$ to be the $j$-th eigenvector of the Laplacian $ \smat L'$. Set $\smash{ \tilde {\vect F}_j^{\rm temp} = \tilde {\svect f}_j^{\rm temp} \otimes \mathbbm 1_N \in \mathbb R^{TN} }$, $j=1,\ldots, T$ using \eqref{eq:mat_Wtemp}  with $N$ vertices and $T$ time slices\;
For each $k$ compute the inner product $\max_{j=1,\ldots,T} |\langle \mathcal{I}(\vect{F}_{k,a}), \tilde{\vect F}_j^{\rm temp} \rangle |$, where $\smash{ \mathcal{I}(\cdot): \mathbb R^{N'} \to \mathbb R^{TN} }$ is the natural inclusion of present vertices\;
Declare $R$ vectors (namely $\{\vect F_{k,a}\}_{k=k_1, k_2 \ldots, k_R}$) with \textit{low} inner product values in the previous step as ``spatial''. 
\end{algorithm}\DecMargin{1em}

With Algorithm \ref{alg:spec_part_findspat}, we construct a version of Algorithm \ref{alg:spec_part} that handles the non-multiplex situation.
\IncMargin{1em}
\begin{algorithm}[h!]
\caption{Spectral partitioning in the non-multiplex case.}\label{alg:spec_part_nonmultiplex}
\SetKwData{Left}{left}\SetKwData{This}{this}\SetKwData{Up}{up}
\SetKwFunction{Union}{Union}\SetKwFunction{FindCompress}{FindCompress}\SetKwFunction{SEBA}{SEBA}
\SetKwInOut{Input}{Input}\SetKwInOut{Output}{Output}
\Input{{Maximum number of vertices in a layer} $N$, number of time slices $T$, graph $\set G$ represented by collection of spatial adjacencies 
$\{ \smat W^{\rm spat}_t\}_{t=1,\dots, T}$, $\smat W^{\rm spat}_t \in \mathbb{R}^{N_t \times N_t}$, $\sum_t N_t =: N' < TN$.} 
\Output{$K$-packing $\{\set{X}_{n_1},\ldots,\set{X}_{n_K}\}$ }
\BlankLine
Construct $\sadjmat$ as described in Sec.~\ref{sec:nonmultiplex} and consequently the corresponding $\slapmat$ using~\eqref{eq:mat_supL}\;
Identify spatial eigenvalues using Algorithm \ref{alg:spec_part_findspat} and compute $a =  a_c$ such that temporal and spatial diffusion is balanced using the heuristic explained in Sec.~\ref{sec:nonmultiplex}\;
 Compute $(\eval_{k,a}, \vect F_{k,a})$, the $k$-th  eigenpair of $\slapmat$ for $k=2,\ldots,R+1$. Select a suitable value for $R$ as described in step 4 of Algorithm \ref{alg:spec_part}\; 
 
Isolate spacetime packing elements: Apply the SEBA algorithm (\cite{Froyland2019}) to the collection $\{\vect F_{k_1,a}, \dots,  \vect F_{k_R,a}\}$. The output SEBA vectors are denoted $\{\vect S^{R}_{j,a}\}_{j=1,\dots,R}$\; 
Identify spurious SEBA vectors:  If the $j^{th}$ SEBA vector is approximately constant on each index block corresponding to a single time fibre, i.e.\ if $[\vect s^{R}_{j,a}]_z\approx C_{j,t}$ for $z=\func i(t,x)$, $1\le x\le N$, $1\le t\le T$, where $\func i$ is defined in \eqref{eq:lexicographic_nonmultiplex}\;
Define packing elements:
Let $\{{n_1}
,\ldots,{n_{K}}\}$, $K\le R$ denote indices of the non-spurious SEBA vectors. Define $\set{X}_{n_j}=\{(t,x) \,:\, [\vect S^{R}_{n_j,a}]_z>0, z=\func i(t,x)\}$ for each $j=1,\ldots,K$.  If a single spacetime vertex $(t,x)$ belongs to two or more $\set{X}_{n_j}$, assign it only to the $\set{X}_{n_j}$ with the largest $[\vect s^{R}_{n_j,a}]_z$ value. One may augment the packing with the unclustered set $\set \Omega = \{(t,x) \,:\, \smash{[\vect s^{R}_{n_j,a}]_z} = 0, z = \func i(t,x)\}$ to obtain a $K+1$ partition.
\end{algorithm}\DecMargin{1em}

\subsection{A toy model with appearing/disappearing vertices}

We consider a non-multiplex network with spacetime nodes $\set{V}$ shown as blue or red in Fig.~\ref{fig:toy_sen} (right).
\begin{figure}[t]
    \centering
    \includegraphics[width=\textwidth]{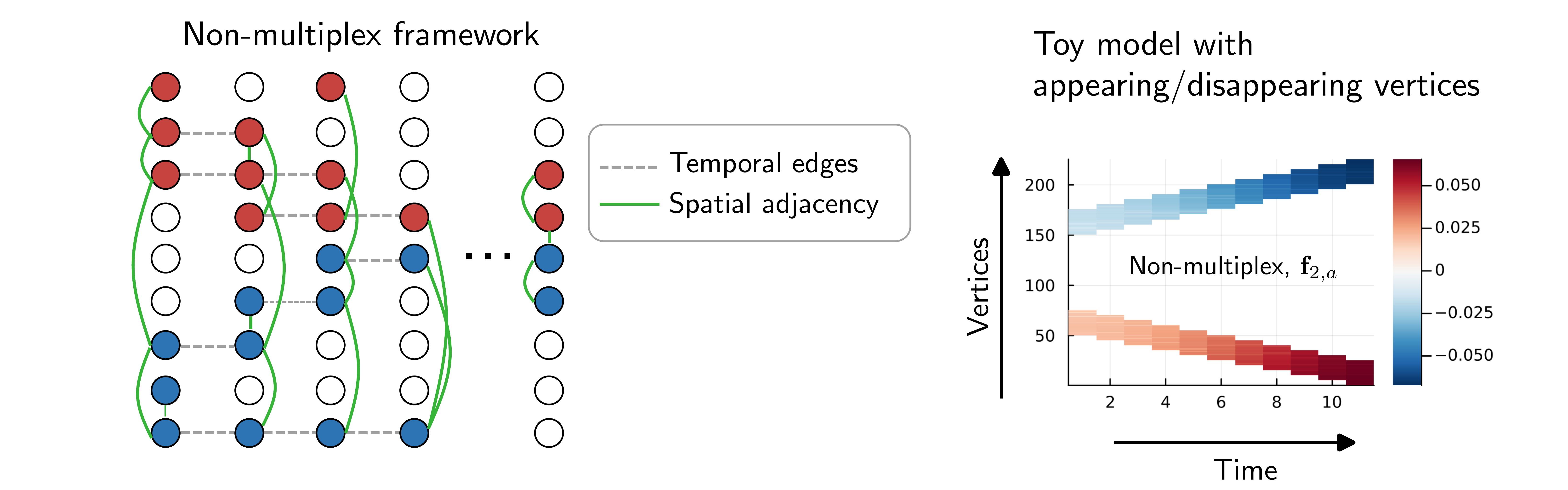}
    \caption{Spectral partitioning on networks with vertices that appear/disappear over time. (Left) Non-multiplex framework, where vertices are connected to their time-analogues only if the vertex is active (red/blue) in the next time slice. (Right) Spectral partitioning on a toy network. A network of 225 vertices and 11 time slices is constructed such that at each time, 5 vertices are exchanged between each {layer} cluster and the inactive vertices. The inter-{layer-}cluster edge weights are set to $1/t$ for $t=1,\dots, T$. The  eigenvector $\unnorm {\vect F}_{2,a}$ corresponding to the unnormalised Laplacian $\ulap$ is shown.}   \label{fig:toy_sen}
\end{figure}
The cardinality of $\hat {\sset V}$ is 225
and there are $T=11$ time slices. For each time slice {$t=1,\ldots,11$}, we assign edge weights that clearly define %have
two time-parameterised families of $16$-regular clusters: $\sset C^{(1)}_t:=\{x \,:\, x=51-5(t-1),\ldots, 75-5(t-1)\}\subset V_t$ and $\sset C^{(2)}_t:=\{x \,:\, x = 151-5(t-1),\dots, 175-5(t-1)\}\subset V_t$ for $t=1,\dots, 11$. 
As $t$ increments to $t+1$, each cluster retains $20$ vertices from the old clusters, loses $5$ vertices that become absent, and gains $5$ newly appearing vertices. 
The intercluster edges are randomly assigned such that a maximum of $20$\%  of the total edges assigned to any present vertex is an intercluster edge. 
The weights $\sset [W^{\rm spat}_t]_{x,y}$ for a pair of present vertices $x$ and $y$ at slice $t$ are %such that
\begin{equation}
    \sset [W^{\rm spat}_t]_{x,y} = \begin{cases}
    1,& \mbox{if $x,y\in \sset C^{(1)}_t$ or $x,y \in \sset C^{(2)}_t$,}\\
    1/t,& \mbox{otherwise.}\\
    \end{cases}
\end{equation}
 We now implement the non-multiplex framework to formulate the corresponding new adjacency matrix $\sadjmat$.

We use Algorithm~\ref{alg:spec_part_nonmultiplex} to perform spectral partitioning.
For our spectral clustering we first fix $a=5.5844$ using step 2 in Algorithm~\ref{alg:spec_part_nonmultiplex} (implementing the heuristic introduced at the end of Sec.~\ref{sec:nonmultiplex}). 
Note that the value of $\eval_{2,a} = \eval_{2,a}^{\rm spat} = 0.29955$. 
{As part of this determination of the value of $a$, in} step 2 {of}  Algorithm~\ref{alg:spec_part_findspat}, we find inner products $\smash{ \max_j |\langle\mathcal I(\vect F_{2,a}), \tilde{\vect F}^{\rm temp}_j \rangle | = 0.00516 }$ and $\smash{\max_j|\langle\mathcal I(\vect F_{3,a}), \tilde{\vect F}^{\rm temp}_j\rangle | = 0.47044}$ and conclude that $\vect F_{2,a}$  corresponds to the first nontrivial spatial eigenvector.
Note that the inner product $\smash{\langle \mathcal I (\vect{F}_{k,a}), \tilde{\vect F}_j^{\rm temp} \rangle}$ is equivalent to computing correlation coefficients $ \rho (\mathcal I(\vect F_{k,a}), \smash{\tilde{\vect F}_j^{\rm temp}}) \in [-1,1] $ for mean-zero and $\ell_2$-normalised vectors. 

The first five nontrivial spatial eigenvalues are $\eval_{2,a}= 0.3029, \eval_{4,a} =  1.3638, \eval_{6,a} = 4.1389, \eval_{8,a}=7.7838$ and $\eval_{10,a} = 10.0706$. As there is no {unambiguous} spectral gap, we look at mean Cheeger {values} $\tilde {\cfunc H}(R)$. For $R=1, 2, 3$ we have the corresponding values $\tilde {\cfunc H}(R) = 0.5969, 2.2080,  4.8186$. Consequently, we choose $R=1$ -- as larger $R$ will create significantly worse Cheeger {ratios} -- and use $\vect F_{2,a}$ to perform spectral clustering. 
In Fig.~\ref{fig:toy_sen} (right), we plot the second eigenvector $\vect F_{2,a}$.  The two persistent clusters are recovered. 
Moreover, the magnitude of $\vect F_{2,a}$ increases with time (represented by darker colours as we move to the right) as the clustering becomes stronger due to the intercluster spatial edges becoming weaker as time increases. This is a vital observation for the following section when we discuss the network of US senators.

\subsection{A network of US senators and states to analyse polarisation}

In this section we use {both} the {multiplex} and non-multiplex approaches on a network of voting similarities in the US Senate between the years 1987 and {2025}. The network represents senators as vertices and has edges weighted by a factor in the range $[0,1]$, which describes the voting similarity between two senators (fraction of identical votes). A higher value indicates that the pair of senators had a {larger} share of common votes on the various bills that they legislated on in the Senate. The network is constructed using details on senators and their votes from \texttt{www.voteview.com}. More details on the construction of the network can be found in~\cite{Waugh2009}.

\subsubsection{A network of senators}

Consider a sequence of weight matrices $\{ \smat W^{\rm spat}_t\}_{t=1,\dots,T}$, $ \smat W^{\rm spat}_t \in \mathbb R^{N_t \times N_t}$ with $N=334$  vertices each  and $T=19$ time slices. 
Each vertex in $V'$ represents a  unique senator that served between the years 1987 and 2025. This corresponds to the congresses numbered in the interval 100--118. The weights $\smat W^{\rm spat}_t$ on the edges between any two vertices describe the voting similarity between the two corresponding senators in congress $99+t$. The weights $W^{\rm spat}_t$ are defined by
\begin{equation}
   [\smat W^{\rm spat}_t]_{x,y} = \frac{1}{|\Omega_{xyt}|}\sum_{i \in \Omega_{xyt}}\delta_{b_{ixt},b_{iyt}} , 
\end{equation}
where $\Omega_{xyt}$ is the set of bills voted on by both senators $x$ and $y$ at congress $99+t$ and $\delta$ is the Kronecker delta function such that $\delta_{ij} = 1 \iff i=j$, $0$ otherwise; thus, the weights $\smash{ [\smat W^{\rm spat}_t]_{x,y} \in [0,1] }$. 
The expression $b_{ixt}$ represents the vote of a senator $x$ for bill $i$ in congress~$99+t$; cf.~\cite{Waugh2009}. The vote is assigned a value of $1$ if voted ``yes'' and assigned $-1$ if voted ``no''. The vote is set to $0$ in case of abstention.  Next, we construct $\sadjmat$ using the non-multiplex framework as discussed previously and formulate the unnormalised supra-Laplacian $\slapmat$. We then use the spectral partitioning Algorithm~\ref{alg:spec_part_nonmultiplex} to detect communities in this network. 

\paragraph{Clustering using $\unnorm {\vect F}_{4,a}$  from the non-multiplex senator network.} 
{Using Algorithm~\ref{alg:spec_part_nonmultiplex}, step 2 we find the value $a=19.90949$ as the critical $a$ for the senator network. This involves the computation of the inner products $\max_j |\langle\mathcal I(\vect F_{2,a}), \smash{\tilde{\vect F}^{\rm temp}_j}\rangle | = 0.54568$, $\max_j |\langle\mathcal I(\vect F_{3,a}), \tilde{\vect F}^{\rm temp}_j\rangle | = 0.50617$ and $\max_j |\langle\mathcal I(\vect F_{4,a}), \tilde{\vect F}^{\rm temp}_j\rangle | = -0.096624$, which identify  $\vect F_{(2,3),a}$ and $\vect F_{4,a}$ as temporal and spatial eigenvectors respectively.}
 {Having fixed the value of $a$}, next, in step 3 of Algorithm \ref{alg:spec_part_nonmultiplex}, we compute several leading eigenvectors $\vect F_{k,a}$ of $\slapmat$, displayed in Fig.~\ref{fig:sen} (second column, upper panel). The first five nontrivial spatial eigenvalues are $\eval_{4,a} = 26.4632, \eval_{6,a} = 35.9543, \eval_{9,a} = 44.3058, \eval_{10,a} = 45.0214$ and  $\eval_{11,a} = 45.1624$.
  As there is significant spectral gap between  $\eval_{4,a}$ and $\eval_{6,a}$, we  choose $R=1$, and use $\vect F_{4,a}$ to perform spectral clustering.
 The eigenvectors shown in Fig.~\ref{fig:sen} (first column, upper panel) have a particular senator ordering for convenience. The senators are sorted by their respective temporal means computed from $\vect F_{4,a}$ (i.e.\ sorted according to average values along horizontal rows in the image). A clear 2-partition is visible, which coincides with the respective party membership of the associated senators, see Fig.~\ref{fig:sen} (first column, upper panel, {plot of $\vect F_{4,a}$}). 

There are two main conclusions one can draw from the plot of the eigenvector $\vect F_{4,a}$ in Fig.~\ref{fig:sen} (upper panel). First, the partitioning of a network based on voting similarities coincides with the partition arising from party memberships--which we know a priori--shown in Fig.~\ref{fig:sen} (upper left), which implies that senators are most likely to vote along party lines. The party affiliation plot in Fig.~\ref{fig:sen} is plotted according to the senator ordering in the eigenvector plots. Second, the norm of $\vect F_{4,a}$ restricted to the 19 individual congresses increases as time passes. This is seen by a darkening of colours from left to right in Fig.~\ref{fig:sen} (right), and corresponds to a higher confidence of partitioning the individual time slices. This can be explained by the inter-cluster edge weights, which decrease as time passes, similar to the toy model from the previous section. In this way we demonstrate growing polarisation between the years 1987 and 2025 in US politics.

\subsubsection{A network of states}
\label{sec:states}
\begin{figure}[htbp]
    \centering
    \includegraphics[width=\textwidth]{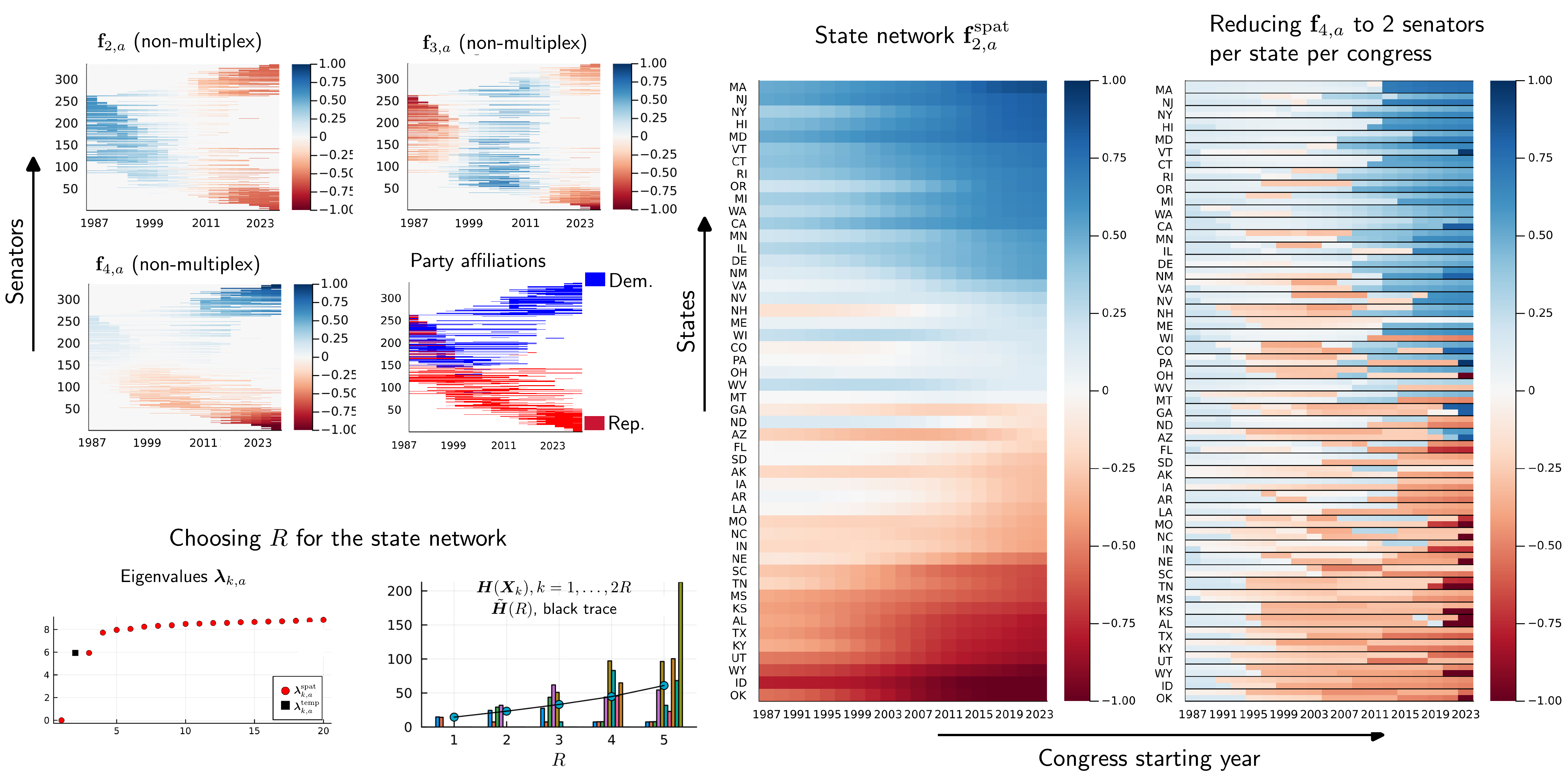}
    \caption{Polarisation in voting patterns of US senators over 19 congresses (numbered 100-118)  between the years 1987--2025. \textit{Upper left:} Spectral partitioning of the network of senators constructed using the non-multiplex approach. Eigenvectors $\vect F_{(2,3,4),a}$ corresponding to the unnormalised Laplacian $\slapmat$ are shown. The vertices in the senator network are ordered according to the temporal means of the spatial fibres $\vect f_{k,a}(\cdot,x)$. Also shown are party affiliations corresponding to senators over time: Democrat (blue) and Republican (red). \textit{Lower left:} Plots of eigenvalues $\eval_{k,a}$ (left) and Cheeger {ratios} $\cfunc H(\set X_k)$ and $\tilde {\cfunc H}(R)$ for various values of $R$ (right) for the state network. \textit{Right:} Spectral partitioning of the network of states discussed in Section \ref{sec:states}. 
    The first spatial eigenvector $\smash{ \vect F^{\rm spat}_{2,a} }$ corresponding to the unnormalised Laplacian $\slapmat$ is shown.  To compare with the state network, we display $\vect F_{4,a}$ from the non-multiplex senator network (upper left) now re-ordered by state affiliation. There are two senators for each state, which are demarcated between different states by black lines.}
    \label{fig:sen}
\end{figure}

We now construct a network with US states as space vertices, and 
spatial edges between two spacetime vertices at time $t$ weighted by the voting similarity of the two corresponding states over the single congress period $99+t$. 
From 1987--2025, the US had 50 states, so this network is of multiplex type.
The network of states is represented by a sequence of adjacency matrices $\{\smat W^{\rm spat}_t\}_{t=1,\dots, T}$, $\smat W^{\rm spat}_t \in \mathbb R^{N \times N}$ with $N=50$ vertices per time slice and $T=19$ time slices. 
We first define the ``aggregate vote'' $b^{\rm state}_{ixt}$  made by a state $x$ on bill $i$ in congress $99+t$, using the individual senator votes $b_{i \sigma t}$ as follows (bill $i$, senator $\sigma$, time~$t$),

\begin{equation}
    b^{\rm state}_{ixt} = \sum_{\sigma \in x} b_{i \sigma t}.
\end{equation}
The weights $[\smat W^{\rm spat}_t]_{x,y}$ are then  given by
\begin{equation}
   [\smat W^{\rm spat}_t]_{x,y} = \frac{1}{|\Omega^{\rm state}_{xyt}|}\sum_{i \in \Omega^{\rm state}_{xyt}}\delta_{b_{ixt}^{\rm state},b_{iyt}^{\rm state}},
\end{equation}
where $\Omega_{xyt}^{\rm state}$ is the set of votes jointly voted on by senators $\sigma$ belonging to states $x$ and $y$ in congress $99+t$, and $\delta$ is the Kronecker delta as before. Consequently, $\smash{ [\smat W^{\rm spat}_t]_{x,y} \in [0,1] }$. The sequence $\{\smat W^{\rm spat}_t\}_{t=1,\dots, T}$ is then used for spectral partitioning using Algorithm~\ref{alg:spec_part} to detect clusters.

As the network is of multiplex type, we can directly use Algorithm~\ref{alg:spec_part}, e.g.\ step 3 to compute $a=14.7532$ and construct $\slapmat$, and step 4 to determine $R$, as discussed in  Sec.~\ref{sec:alg}. 
Regarding the choice of $R$, in Fig.~\ref{fig:sen} (lower left) we see that the main nontrivial gap is between the 2nd and 3rd spatial eigenvalues, suggesting $R=1$.
 Fig.~\ref{fig:sen} (lower left) shows the plot of the average Cheeger {ratio}~$\tilde {\cfunc H}(R)$. The average Cheeger {ratio} {increases} consistently  with increasing~$R$ and there does not seem to be a possibility to form more than two clusters with low Cheeger value. For this reason, and the above spectral gap considerations, we fix~$R=1$. 

In Fig.~\ref{fig:sen} (right) we plot the first nontrivial spatial eigenvector $\smash{ \vect F^{\rm spat}_{2,a} }$ (the second spatial eigenvector) of $\slapmat$ computed from the state network. The vertices are labelled with their respective state abbreviations and sorted according to the temporal average of the displayed eigenvector (i.e.\ sorted according to average values along horizontal rows in the image). 

\subsubsection{Comparing state and senator networks}

To compare the results we have obtained for the senator and state networks, we restrict the number of senators in a single congress to exactly two by selecting the two senators who contributed the most votes in that congress. 
After this restriction, for a single congress we have 100 senators across the 50 states, and over the 19 congresses in total we have 1900 senator/congress pairs.
These 1900 senator/congress pairs correspond to subsampling 1900 values from $\vect F_{4,a}$ and these values are mapped in the obvious way into Fig.~\ref{fig:sen} (far right), which can now be directly compared with Fig.~\ref{fig:sen} (right).

In Fig.~\ref{fig:sen} (right) we observe that the blue/red partitioning obtained from the senator network and state network is broadly consistent. Both plots show a deepening of colours as time progresses, demonstrating a temporally growing political polarisation between the years 1987 and 2025. This is consistent with recent studies, cf.~\cite{Kleinfeld_2023}.   In \cite{Mucha2010}, this was studied without a
representation of the level of confidence of partitioning. 
Our spectral partitioning technique allows for unclustered vertices, provides information about the strength of the clusters, and can track evolution of both cluster elements and cluster strength over time, all of which are important to
examine growing polarisation.

\section{Conclusion}

In this work we made several important contributions towards a more rigorous theory for spectral clustering of time-varying complex networks and formulated a versatile algorithm for both multiplex and non-multiplex cases.
Our {key construction} is the supra-Laplacian $\slap$, {which} balances spatial and temporal diffusion {through} a scalar diffusion constant $a$.
{We constructed both unnormalised and normalised supra-Laplacians on graphs and stated corresponding spacetime Cheeger inequalities.}
{We fully explored the hyperdiffusivity limit $a\to\infty$ in both the unnormalised and normalised cases, proving convergence of of the spectrum and eigenvectors of $\slap$ and $\smash{\nlap}$ to the dynamic Laplacian $\dlap$ and normalised temporal Laplacians $\ltemp$, respectively.}

This theory informed the construction of novel spectral clustering algorithms, which automatically identified the eigenvectors that carry cluster-related information, and applied sparse eigenbasis approximation (SEBA) to isolate individual spacetime clusters.
{Our spectral clustering algorithms are designed for time-varying networks with (i) constant-in-time vertex sets {on the network's layers} (multiplex), and  (ii) vertex sets {that may vary across the network's layers} (non-multiplex).}
{We demonstrated that our novel spectral approach is} effective in detecting {and tracking several communities at once}, {as well as} their gradual appearance and disappearance, in {both} multiplex and non-multiplex {situations}. 
The algorithms are conveniently implemented in \url{https://github.com/mkalia94/TemporalNetworks.jl}.

\subsection*{Acknowledgements}
The research of GF was supported by an Einstein Visiting Fellowship funded by the Einstein Foundation Berlin, an ARC Discovery Project (DP210100357) and an ARC Laureate Fellowship (FL230100088).
GF is grateful for generous hospitality at the Department of Mathematics, University of Bayreuth, during research visits.
GF thanks Kathrin Padberg-Gehle for helpful conversations during a visit in 2018. MK is thankful to Nataša Djurdjevac Conrad for informative discussions on generation of temporal networks and spectral clustering.
PK acknowledges support by the Deutsche Forschungsgemeinschaft (DFG, German Research Foundation) -- 546032594.

\newpage

\section*{Appendix}
\begin{appendices}

\section{Proof of Proposition~\ref{prop:h}}
\label{sec:prop_h_proof}
\begin{proof}
First, we prove the result $\func h_K \leq \func h_{K+1}$ by showing that from any $(K+1)$-packing we obtain a $K$-packing with a maximal Cheeger ratio not larger than that of the $(K+1)$-packing we started with. The case for the normalised Cheeger constant $\norm {\func h}_K$ follows identically. For a given graph $\set G$, consider the optimal $(K+1)$-packing
$\partn X$ of $\set V$ given by $\partn X=\{ \set X_1, \ldots, \set X_K, \set X_{K+1} \}$. Thus the Cheeger constant $\func h_{K+1}$ is given by,
\begin{equation}
    \func h_{K+1}= \min_{\substack{\{\set Y_1, \ldots, \set Y_{K+1}\} \\ \mbox{ is a packing of } \set V}} \max_k \cfunc H(\set Y_k) = \max_{k=1,\ldots,K+1} \cfunc H(\set X_k).
\end{equation}
Without loss of generality, let $ \max_{k=1,\ldots,K+1} \cfunc H(\set X_k) = \cfunc H(\set X_{K+1})$. Then we have that 
\begin{equation}
    \max_{\substack{k=1,\ldots,K \\ \set X_k \in \hat{\partn X}}} \cfunc H(\set X_k) \le \cfunc H(\set X_{K+1}).
    \label{eq:prop2-3}
\end{equation}
% $ \cfunc H(\set X_k) \le \cfunc H(\set X_{K+1}) $ for $k \ne K+1$. 
Let us define a new $K$-partition $\hat{\partn X} = \{\set X_1,\ldots,\set X_{K-1}, \linebreak[3] \set X_{K} \cup \set X_{K+1} \}$. We claim that
\begin{equation}
    \begin{aligned}
        \cfunc H(\set X_K \cup \set X_{K+1}) \le \cfunc H(\set X_{K+1}).
    \end{aligned}\label{eq:prop2-1}
\end{equation}
Note that $\func \sigma\left(\set X_K \cup \set X_{K+1} \right ) \le \cfunc \sigma(\set X_K) + \cfunc \sigma(\set X_{K+1})$ and that $|\set X_K \cup \set X_{K+1}| =  |\set X_{K}| + |\set X_{K+1}|$, as $\set X_K$ and $\set X_{K+1}$ are disjoint.
Thus we have
\begin{equation}
    \cfunc H(\set X_K \cup \set X_{K+1}) \le \frac{\cfunc \sigma(\set X_K) + \cfunc \sigma(\set X_{K+1})}{|\set X_{K}| + |\set X_{K+1}|} \le \frac{\cfunc \sigma(\set X_{K+1})}{|\set X_{K+1}|} = \cfunc H(\set X_{K+1}),
    \label{eq:prop2-2}
\end{equation}
where the second inequality comes from the fact that if  $\frac{a}{b} \le \frac{c}{d}$ then {$\frac{a+c}{b+d} \le \frac{c}{d}$.} % $\frac{a+b}{c+d} < \frac{c}{d}$.
This proves~\eqref{eq:prop2-1}. 
% From \eqref{eq:prop2-2} we also have that
Now,
\begin{equation}
    \cfunc h_K = \!\!\! \min_{\substack{\{\set Y_1, \ldots, \set Y_K\} \\ \mbox{ is a packing of } \mathcal V}} \max_{k=1, \dots, K} \cfunc H(\set Y_k) \le \max_{\substack{k=1,\ldots,K \\ \set X_k \in \hat{\partn X}}} \cfunc H(\set X_k) \overset{\eqref{eq:prop2-3}}{\le} \cfunc H(\set X_{K+1}) = % \max_{\substack{ k=1, \ldots, K+1 \\ \set X_k \in \partn X}} \cfunc H(\set X_k) = 
    \cfunc h_{K+1}.
\end{equation}
To prove the result $\norm {\cfunc h}_K \leq \norm {\cfunc h}_{K+1}$, one simply needs to replace node counts $|\cdot|$ by the degree $\adep {\func d}(\cdot)$ and all preceding arguments hold.
\end{proof}

\section{Proof of Lemma \ref{lem:spattemp}}
\label{sec:lem:spattemp}
\begin{proof}

(Part 1) Consider the expression $\ulspat \unnorm {\func F}_{k,a}^{\rm temp}(t,x)$. Expanding using \eqref{eq:lspat_ltemp} we get {for all $1\le t\le T$ and $x\in \sset V$} that
\begin{align*}
    \ulspat\unnorm {\func F}_{k,a}^{\rm temp}(t,x)
    &=  \sum_{y}\adjspat_{(t,x),(t,y)}\left( \unnorm {\func F}_{k,a}^{\rm temp}(t,x) - \unnorm {\func F}_{k,a}^{\rm temp}(t,y)\right) \\
    &=  \sum_{y}\adjspat_{(t,x),(t,y)}\left( \unnorm {\sfunc f}_k^{\rm temp}(t) - \unnorm {\sfunc f}_k^{\rm temp}(t)\right) = 0.
    \label{eq:lemma7-1}
    \stepcounter{equation}\tag{\theequation}
\end{align*}
Similarly expanding $\ultemp \unnorm {\func F}^{\rm temp}_{k,a}(t,x)$ using \eqref{eq:lspat_ltemp} we get for all $1\le t\le T$ and $x\in V$ that
\begin{align*}
   \ultemp \unnorm {\func F}^{\rm temp}_{k,a}(t,x) &= \sum_{s}\adjtemp_{(t,x),(s,x)}\left( \unnorm {\func F}_{k,a}^{\rm temp}(t,x) - \unnorm {\func F}_{k,a}^{\rm temp}(s,x)\right) \\
   & \stackrel{\eqref{eq:Wtemp_op}}{=} \sum_s \sop{W}'_{t,s} (\unnorm {\func F}_{k,a}^{\rm temp}(t,x) - \unnorm {\func F}_{k,a}^{\rm temp}(s,x)) \\ 
   &= \sum_s \sop{W}'_{t,s} (\sfunc f_k^{\rm temp}(t) - \sfunc f_k^{\rm temp}(s)) \\
   &= \sop{L}' \sfunc f_k^{\rm temp}(t) \\
   &= \seval_k^{\rm temp}\sfunc f_k^{\rm temp}(t) \\
   &= \seval_k^{\rm temp}\unnorm {\func F}_{k,a}^{\rm temp}(t,x).
   \label{eq:lemma7-2}\stepcounter{equation}\tag{\theequation}
\end{align*}
From \eqref{eq:lemma7-1} and \eqref{eq:lemma7-2} we get
\begin{equation}
    \ulap \unnorm {\func F}_{k,a}^{\rm temp}(t,x) = \ulspat \unnorm {\func F}_{k,a}^{\rm temp}(t,x) + a^2 \ultemp \unnorm {\func F}_{k,a}^{\rm temp}(t,x) = \seval_k^{\rm temp}\unnorm {\func F}_{k,a}^{\rm temp}(t,x).
\end{equation}
(Part 2) Consider any nontemporal %other
eigenpair of~$\ulap$. For its eigenvector $\unnorm {\func F}_{j,a}(t,x)$ we have due to self-adjointness of $\ulap$ that
\begin{equation}
\langle \unnorm {\func F}_{j,a},\unnorm {\func F}_{k,a}^{\rm temp} \rangle = 0, \quad \text{for all } k\ge 2,
\end{equation}
which gives
\begin{equation}
    \sfunc f_k^{\rm temp}(1)\sum_x \unnorm {\func F}_{j,a}(1,x) + \sfunc f_k^{\rm temp}(2)\sum_x \unnorm {\func F}_{j,a}(2,x)  + \dots + \sfunc f_k^{\rm temp}(T) \sum_x \unnorm {\func F}_{j,a}(T,x) = 0.
\end{equation}
Consequently, the function $\bar {\sfunc f}_j$ of slicewise sums, $\bar {\sfunc f}_j(t) := \sum_x \unnorm{\func F}_{j,a}(t,x)$, is perpendicular to all temporal functions $\sfunc f^{\rm temp}_k$, $k\ge 2$. As $\mathrm{span}\left\{ \sfunc f^{\rm temp}_k \,:\, k\ge 2 \right\} = \mathbbm{1}^{\perp} \subset \mathbb{R}^T$, we obtain that $\bar {\sfunc f}_j \in \mathrm{span} \{\mathbbm{1} \}$, implying
% As $\sum_s f^{\rm temp}(s) = 0$ ($\mathcal L' \mathbbm{1} = 0$ and $f^{\rm temp} \perp \mathbbm{1}$), we must have
\begin{equation}
    \sum_x \unnorm {\func F}_{j,a}(1,x) = \sum_x \unnorm {\func F}_{j,a}(2,x) = \ldots = \sum_x \unnorm {\func F}_{j,a}(T,x) = C,
    \label{eq:lemma7-3}
\end{equation}
for some constant $C$. The eigenfunctions $\func F_{j,a} \equiv \func F^{\rm spat}_{\ell,a}$ are called spatial eigenfunctions, while the possible difference in the indices $j,\ell$ is due to the fact that while $j$ refers to global enumeration of eigenvectors (in ascending order), $\ell$ enumerates only the spatial modes. 
\end{proof}

\section{Proof of Theorem \ref{thm:var}}
\label{sec:them:var}
We begin by defining spatial and temporal eigenspaces of the supra-Laplacian~$\ulap$. 
\begin{definition}{\bf (Spatial and temporal eigenspaces of $ \ulap$)}
The temporal eigenspace $\mathbb{S}^{\rm temp}$ is defined by
\begin{equation}
    \mathbb{S}^{\rm temp} = {\rm span} \left\{\sfunc f_k^{\rm temp}\otimes \mathbbm{1}_{N} \,:\, \lap' \sfunc f_k^{\rm temp} = \seval_k^{\rm temp}\sfunc f_k^{\rm temp},\ k\ge 2 \right\}. 
\end{equation}
    The spatial eigenspace $\mathbb{S}^{\rm spat}$ is its orthogonal complement:
\begin{equation}
\mathbb{S}^{\rm spat}= (\mathbb{S}^{\rm temp})^\perp.
\end{equation}
\end{definition}

We also note that it is a standard computation~\cite[Sec.~2.1]{Grigoryan2018} to obtain from the definition
\[
\op{L} \func f(t,x) = \sum_{s,y} \op{\adj}_{(t,x),(s,y)} \left(\func f(t,x)-\func f(s,y)\right)
\]
of a (spacetime) graph Laplacian $\op{L}$ associated with a general weight operator $\op{\adj}$ the useful ``variational'' expression (or Green's formula)
\begin{equation}
    \label{eq:Lap_variational}
    \langle\op{L} \func f, \func f \rangle = \frac12  \sum_{s,t,x,y} \op{\adj}_{(t,x),(s,y)} \left(\func f(t,x)-\func f(s,y)\right)^2.
\end{equation}

\begin{proof}[{Proof of Theorem~\ref{thm:var}}]
\,

{Part 1:}
The arguments presented here are analogous to the proof of~\cite[Theorem~7]{FrKo23}. Consider a function $\func F_i(t,x) = \sfunc g_i^D(x)$ where $\sfunc g_i^D(x)$ is the $i$-th eigenfunction of~$\dlap$. Define $\mathbb{S}':={\rm span}\{\func F_1,  \dots , \func F_N\}$ by the span of all such functions $\func F_i$. Then $\mathbb{S}' \subset \mathbb{S}^{\rm spat}$. Further define 
$\mathbb{S}_k':= \textnormal{span}\{\func F_1, \dots , \func F_k\}$ and $\mathbb{S}^{\rm spat}_{k,a} := {\rm span}\{\unnorm {\func F}_{1,a}^{\rm spat}, \dots , \unnorm {\func F}_{k,a}^{\rm spat}\}$, where we take all spatial eigenfunctions to have (Euclidean) norm~1. Now, for spatial eigenvalues $\smash{\unnorm \eval_{k,a}^{\rm spat}}$ and $k=1,\dots , N$, the Courant--Fischer minmax principle and \eqref{eq:Lap_variational} with $\op{L} = \ulap = \ulspat + a^2\ultemp$ give
\begin{align}
    \unnorm \eval_{k,a}^{\rm spat} &= \min\limits_{\substack{{\rm dim }\mathbb{S}=k \\ \mathbb{S} \subset \mathbb{S}^{\rm spat} }} \max_{\substack{\func F \in \mathbb{S} \\ \func F\neq 0}} \ \ \frac{1}{2\sum\limits_{t,x} \func F(t,x)^2} \left( \sum\limits_{t,x,y} (\func F(t,x) - \func F(t,y))^2\adjspat_{(t,x),(t,y)}\right. & \\ &\qquad \qquad + \left. \sum\limits_{s,t,x} a^2(\func F(t,x) - \func F(s,x))^2\adjtemp_{(t,x),(s,x)}\right ) & \label{eq:CourantFischer_spat}\tag{\theequation}\\
    & \leq \max\limits_{\func F \in \mathbb{S}_k' } \ \ \frac{1}{{ 2\sum\limits_{t,x} \func F(t,x)^2}} \left( \sum \sum\limits_{ t,x,y} (\func F(t,x) - \func F(t,y))^2 \adjspat_{(t,x),(t,y)} \right. & \\ &\qquad \qquad + \left. \sum\limits_{s,t,x} a^2\big( \overbrace{\func F(t,x) - \func F(s,x)}^{=0} \big)^2 \, \adjtemp_{(t,x),(s,x)} \right) &\\
    & = \max\limits_{\func F \equiv [\sfunc g, \dots, \sfunc g] \in \mathbb{S}_k' } \frac{\sum\limits_{x, y}(\sfunc g(x) - \sfunc g(y))^2 \sum_t \adjspat_{(t,x),(t,y)}}{ 2 N \sum\limits_{x}  \sfunc g(x)^2} \\
    & = \max\limits_{\func F \equiv [\sfunc g, \dots, \sfunc g] \in \mathbb{S}_k' } \frac{\sum\limits_{x, y}(\sfunc g(x) - \sfunc g(y))^2 \dadj_{x,y} }{ 2 \sum\limits_{x}  \sfunc g(x)^2} = \seval_k^D. &
\label{eq:bound}\stepcounter{equation}\tag{\theequation}
\end{align}
From \eqref{eq:order} and~\eqref{eq:bound} we have 
\begin{equation}
    \unnorm \eval_{k,a} \leq \seval_k^D, \ \forall a.
\end{equation}
This proves statement 1.

{Part 2:} To show that $\smash{a \mapsto \unnorm \eval_{k,a}^{\rm spat} }$ is nondecreasing for $1\le k\le TN$, consider again \eqref{eq:CourantFischer_spat}, the Courant--Fischer characterisation of~$\smash{ \eval_{k,a}^{\rm spat} }$. Note that for fixed $\mathbb{S}$ and fixed $\func F\in\mathbb{S}$ the objective function is nondecreasing in~$a$. Thus, keeping $\mathbb{S}$ fixed, but maximising over $\func F\in\mathbb{S}$ still gives an expression nondecreasing in~$a$. Since this holds for any $\mathbb{S}$, taking finally the minimum over these subspaces yields the claim. As $\eval_{k,a}^{\rm temp} = a^2 \seval_{k}^{\rm temp}$ is monotonically increasing with respect to $a$, statement 2 follows. 

{Part 3:}
Because $a \mapsto \unnorm \eval_{k,a}^{\rm spat}$ is nondecreasing, we have that
\begin{equation}
    \lim_{a \to \infty}\unnorm \eval_{k,a} \leq \lim_{a \to \infty} \unnorm \eval_{k,a}^{\rm spat} \leq \lambda_{k}^D
    \label{eq:ineq1}
\end{equation}
for $1\le k\le N$.
This also implies $\lim_{a \to \infty}\unnorm \eval_{k,a} = \lim_{a \to \infty} \unnorm \eval_{k,a}^{\rm spat}$, because temporal eigenvalues grow indefinitely with $a$, and thus for sufficiently large $a$ the $k$-th eigenvalue is spatial.

For statement 3, we are left to show that $\smash{ \seval^D_k \le \lim_{a \to \infty} \unnorm \eval_{k,a}^{\rm spat} }$.
Ahead, we fix $k \in \{1, \dots, N\}$.  Consider the set $\bar{\mathbb{S}}_k:= \{ F \in \mathbb{S}_k^{\rm spat} \,:\, \|F\|_2 = 1\}$. As $\bar{\mathbb{S}}_k$ is compact in $\mathbb{R}^{NT}$, there is a subsequence $a_i \to \infty$ as $i \to \infty$, such that  minimisers $\smash{ \unnorm {\func F}_{k,a_i}^{\rm spat} }$ of \eqref{eq:CourantFischer_spat} have a limit $\func F_k^*$, which also lies in~$\bar{\mathbb{S}}_k$. 

We now claim that $\func F_k^* \in \mathbb{S}'$, or, equivalently, that~$\ultemp \func F_k^* (t,x) = 0$. To this end we note that because $\|\smash{ \unnorm {\func F}^{\rm spat}_{k,a} }\|_2=1$,  %\gf{normalization is not used in the next display eqn} %\pk{If they weren't normalized, I'd have needed to divide by $\|f^{\rm spat}_{k,a_i}\|^2$.} 
\begin{equation*}
    \infty > \lambda_k^D \geq \langle \unnorm {\op{L}}^{(a_i)}\unnorm {\func F}_{k,a_i}^{\rm spat}, \unnorm {\func F}_{k,a_i}^{\rm spat} \rangle \geq a^2 \langle \ultemp \func F_{k,a_i}^{\rm spat}, \func F_{k,a_i}^{\rm spat} \rangle.
\end{equation*}
As $a_i \to \infty$, we obtain that 
\begin{equation}
    \label{eq:F_const_in_time}
    \lambda_k^D \ge \limsup_{i\to\infty} a_i^2  \langle \ultemp \func F_{k,a_i}^{\rm spat}, \func F_{k,a_i}^{\rm spat} \rangle = \limsup_{i\to\infty} a_i^2 \langle \ultemp \func F_{k}^{*}, \func F_{k}^{*} \rangle.
\end{equation}
Thus~$\langle \ultemp \func F_{k}^{*}, \func F_{k}^{*} \rangle = 0$, implying by \eqref{eq:Lap_variational} with $\op{L} = \ultemp$ temporal constancy of $\func F_{k}^{*}$, and hence we may write $\func F_k^*$ in the following canonical form:
\begin{equation}
\label{eq:gkstar}
    \func F_k^*(t,x) = \sfunc g_k^*(x), \textnormal{ for some } \sfunc g_k^*: V \to \mathbb{R}^N.
\end{equation}
% In particular, also $\smash{ \mathcal{L}^{\rm temp}f_k^*= 0 }$. 
Recalling that~$\|\func F_k^*\|_2=1$, \eqref{eq:gkstar} yields $\|\sfunc g_k^*\|_2^2 = \frac1T$. 

Let $\ulspat_t$ be the operator that acts on a spacetime function $\func F(t,x)$ as
\begin{equation}
    \ulspat_t \func F(t,x) = \sum_y \adjspat_{(t,x),(t,y)}\left(\func F(t,x)-\func F(t,y)\right).
\end{equation}
It can be interpreted as the spatial Laplacian at time $t$. Then we have the following property for any spacetime function $\func F(t,x)$ and for any~$a$:
\begin{equation}
\label{eq:split_inner_prod}
    \langle \ulap \func F, \func F \rangle = \sum_t  \langle \ulspat_t \func F(t,\cdot), \func  F(t,\cdot) \rangle  + a^2 \langle \ultemp \func F, \func F \rangle \geq \sum_t \langle \ulspat_t \func F(t,\cdot), \func F(t,\cdot) \rangle.
\end{equation}
Moreover,  for an arbitrary sequence of real-valued functions $(h_n)_n$ with pointwise limit $h$ one has\footnote{To see this, for any $z$ and $\epsilon>0$ we denote by $z_\epsilon$ a point with $h(z_\epsilon) \ge \sup_z h(z) - \epsilon$. Then one has $\sup_z h(z) - \epsilon \le h(z_\epsilon) = \limsup_n h_n(z_\epsilon) \le \limsup_n \sup_z h_n(z)$.}
\begin{equation}
    \label{eq:limsup_max_swap}
    \limsup_{n\to\infty} \sup_z h_n(z) \ge \sup_z h(z).
\end{equation}
Recalling that $\lim_{i\to\infty} \unnorm {\func F}^{\rm spat}_{\ell,a_i}(t,x) =  \sfunc g^*_\ell (t,x)$ %\gf{Isn't $f^{\rm spat}_{\ell,a_i}$ a spacetime function but $g^*_\ell$ a space function?} \pk{thanks, corrected}
with $\| \sfunc g^*_\ell \| = \frac1T$, we obtain that
\begin{align*}
    \lim_{a\to\infty} \unnorm \eval_{k,a}^{\rm spat} &= \limsup_{i \to \infty} \unnorm \eval_{k,a_i}^{\rm spat} = \limsup_{i \to \infty}  \max_{\substack{\func f \in {\rm span}\{ \unnorm {\func F}_{1,a_i}^{\rm spat} \dots \unnorm {\func F}_{k,a_i}^{\rm spat}\} \\ \|\func F\|=1}} \langle \unnorm{\op{L}}^{(a_i)} \func  F, \func F \rangle \\
    &\stackrel{\eqref{eq:split_inner_prod}}{\ge} \limsup_{i \to \infty}  \max_{\substack{\func F \in {\rm span}\{ \unnorm {\func F}_{1,a_i}^{\rm spat} \dots \unnorm {\func F}_{k,a_i}^{\rm spat}\} \\ \|\func F\|=1}} \sum_t \langle \ulspat_t \func F(t,\cdot), \func F(t,\cdot) \rangle \\
    &= \limsup_{i \to \infty} \max_{\svect c\in\mathbb{R}^k,\ \|\svect c\|=1}  \sum_t \Big\langle \ulspat_t \sum_\ell \svect c_\ell \unnorm {\func F}^{\rm spat}_{\ell,a_i}(t,\cdot), \sum_\ell \svect c_\ell \unnorm {\func F}^{\rm spat}_{\ell,a_i}(t,\cdot) \Big\rangle \\
    &\stackrel{\eqref{eq:limsup_max_swap}}{\ge} \max_{\svect c\in\mathbb{R}^k,\ \|\svect c\|=1}  \sum_t \Big\langle \ulspat_t \sum_\ell \svect c_\ell \sfunc g^*_\ell, \sum_\ell \svect c_\ell \sfunc g^*_\ell \Big\rangle \\
    &= \max_{\svect c\in\mathbb{R}^k,\ \|\svect c\|=1} \Big\langle \underbrace{\left(\frac1T \sum_t \ulspat_t \right)}_{=\dlap} \sum_\ell \svect c_\ell \sqrt{T} \sfunc g^*_\ell, \sum_\ell \svect c_\ell \sqrt{T} \sfunc g^*_\ell \Big\rangle \\
    &\ge \min_{\substack{S\subset\mathbb{R}^N\\ \dim S=k}} \max_{\substack{\sfunc g\in S\\ \|\sfunc g\|=1}} \left\langle \sop{L}^D \sfunc g,\sfunc g \right\rangle = \seval^D_k,
\end{align*}
the last equality invoking again the Courant--Fischer minmax principle.
This together with \eqref{eq:ineq1} implies the claim of statement~3.

{Part 4:}
Finally, to show in statement 4 that the only accumulation points of the sequence $(\func F^{\rm spat}_{k,a})_{a>0}$ for $1\le k\le TN-T+1$ as $a\to\infty$ are vectors constant in time, assume the contrary. That is, there is a subsequence $(a_j)_{j\in\mathbb{N}}$ with $\lim_{j\to\infty} a_j = \infty$ such that $\smash{\func  F^{\rm spat}_{k,a_j} }$ converges to some $\func F_k^{**}$ satisfying~$\langle \ultemp \func F_{k}^{**}, \func F_{k}^{**} \rangle > 0$ by~\eqref{eq:Lap_variational}. This immediately contradicts~\eqref{eq:F_const_in_time}.
\end{proof}

\section{Proof of Theorem \ref{thm:normL}}
\label{sec:thm:normL}

\begin{proof}
Consider the graph $\set G = (\set V, \set E, \sadj)$, the normalised Laplacian~$\smash{\nlap}$, and define by $\mathcal H$ the space of all functions $\func f: \set V \to \mathbb R$ acting on vertices~$\set V$. We first claim that the operator $\smash{\nlap}$ corresponding to the supra-Laplacian with respect to the graph $\set G$ acts on any arbitrary function $\func F \in \mathcal H$ as
\begin{equation}
    \lim_{a \to \infty} \nlap \func F = \norm \ltemp \func F,
    \label{eq:normL_claim1}
\end{equation}
where $\norm \ltemp$ is the normalised Laplacian defined with respect to the graph $\set G = (\set V, \set E, \linebreak[3] \adjtemp)$.  
Let us denote the spatial and temporal components of $\smash{\nlap}$ by $\smash{\nlspat}$ and $\smash{\nltemp}$, respectively. Thus,
    \begin{equation}
        \nlap  = \nlspat + a^2 \nltemp.
        \label{eq:norml_split}
    \end{equation}
    Using \eqref{eq:deg}, we note that
    \begin{equation}
        \lim_{a \to \infty} \adep{\func d}(t,x) = \lim_{a \to \infty} \sum_{y} \adjspat_{(t,x),(t,y)} + a^2 \sum_{s} \adjtemp_{(t,x),(s,x)} = \infty.
        \label{eq:thm10_1}
    \end{equation}
    Consider the expression $\nlspat \func F$. Expanding using \eqref{eq:generalLap}, we get
    \begin{align*}
       \nlspat \func F &= \sum_{s,y} \adjspat_{(t,x),(s,y)} \left( \frac{\func F(t,x)}{\adep{\func d}(t,x)} - \frac{\func F(s,y)}{\adep{\func d}(t,x)^{\frac{1}{2}}\adep{\func d}(s,y)^{\frac{1}{2}}} \right) \\
        & \to 0 \quad (\textnormal{as } a\to \infty),
        \label{eq:thm10-3}\stepcounter{equation}\tag{\theequation}
    \end{align*}
    which is obtained as $\adep{\func d}(t,x) \to \infty$ for $a \to \infty$ from~\eqref{eq:thm10_1}. Further we note that the limit
\begin{align*}
    \lim_{a \to \infty} \frac{a^2}{\adep{\func d}(t,x)} &= \lim_{a \to \infty} \frac{a^2}{\sum_{s,y} \sadj_{(t,x),(s,y)}} \\ 
    & = \lim_{a \to \infty} \frac{a^2}{\sum_{y}  \adjspat_{(t,x),(t,y)} + a^2 \sum_{s} \adjtemp_{(t,x),(s,x)}} \\
     &= \lim_{a \to \infty} \frac{1}{\sum_{y} \frac{1}{a^2}\adjspat_{(t,x),(t,y)} + \sum_{s} \adjtemp_{(t,x),(s,x)}} \\
    &= \frac{1}{\sum_{s} \adjtemp_{(t,x),(s,x)}} \\
    &= \frac{1}{\sum_s\adj'_{s,t}} = \frac{1}{\sfunc d^{\rm temp}(t)}.
\stepcounter{equation}\tag{\theequation}\label{eq:thm10_2}
\end{align*}
    Now consider the expression $a^2 \nltemp \func F.$ Expanding using \eqref{eq:generalLap},  we get
    \begin{align*}
        a^2 \nltemp \func F &= \sum_{s,y} a^2\adjtemp_{(t,x),(s,y)} \left( \frac{\func F(t,x)}{\adep{\func d}(t,x)} - \frac{\func F(s,y)}{\adep{\func d}(t,x)^{\frac{1}{2}}\adep{\func d}(s,y)^{\frac{1}{2}}} \right) & \\
        &\to \sum_{s,y} \adjtemp_{(t,x),(s,y)} \left( \frac{\func F(t,x)}{\sfunc d^{\rm temp}(t)} - \frac{\func F(s,y)}{\sfunc d^{\rm temp}(t)^{\frac{1}{2}}\sfunc d^{\rm temp}(s)^{\frac{1}{2}}} \right) & (\textnormal{as } a\to \infty) \\
        &= \norm \ltemp \func F,
        \label{eq:thm10-4}\stepcounter{equation}\tag{\theequation}
    \end{align*}
    where the limit is obtained from~\eqref{eq:thm10_2}.
    Now, the claim \eqref{eq:normL_claim1} follows from combining \eqref{eq:norml_split} with \eqref{eq:thm10-3} and \eqref{eq:thm10-4}. As $\mathcal H$ is finite-dimensional, from Dini's theorem we have that
    \begin{equation}
    \label{eq:normLap_a_Limit}
        \nlap \to \norm \ltemp \textnormal{ uniformly as } a \to \infty. 
    \end{equation}
    Finally from \cite[Thm.~II.5.1]{Kato1995} we know that an arbitrary continuous perturbation of $\smash{\norm \ltemp}$ (in the uniform topology) yields eigenvalues and eigenspaces that are arbitrarily close to those of~$\smash{\norm \ltemp}$.
    Thus, we have that the eigenvalues $\norm \eval_{k,a}$ can only have accumulation points $\mu_j$ for some $j=1,\ldots,T$ (and their corresponding eigenspaces accordingly).     
    \end{proof}

\section{Graph generation}
\label{sec:graphgen}
We lay out an empirical graph generation method that incorporates graph transitions explicitly. We begin with a spacetime graph $\set{G}$ with $N$ spatial nodes, $T$ time fibers. We define by $\alpha = \{\alpha_1, \alpha_2, \dots\}$ and $s = \{s_1, s_2, \dots \}$,  $s_j, \alpha_i \in \mathbb{R}$ a sequence of states $\alpha_i$ that represents the number of communities/partitions we would like to have at times $s_i \in [0,T]$ respectively. For instance, $\alpha = \{ 1,0,1 \}$ and $s = \{0, s_1, T\}$ represents a network with a single cluster at initial time $t=0$  that 2-disappears to a network state with no clear partitions at time $t=s_1$, and transitions back to a network state with a single appearing cluster (2-appearance) at time~$t=T$. We also require the scalar parameters $\eta, \beta >0$ and $\gamma \in \mathbb N$. Their role will be made clear after the procedure is explained.

Given a sequence $\alpha$ and the corresponding times $s$, the network $\set{G}$ represented by a sequence of adjacencies $\{\smat W^{(t)}\}_{t=1,\dots, T}$  -- which will constitute the spatial weight matrices $\smat W^{\rm spat}_t$ on the individual time slices, having only entries 0 and~1 --  is generated as follows:
\begin{enumerate}
    \item For every $\alpha_i \in \alpha$, construct a time slice/network layer as follows:
    \begin{enumerate}
        \item Generate $\alpha_i$ fully connected networks of size $\lfloor N/(\alpha_i+\beta) \rfloor $ each, for some $\beta$. These represent clusters. This still leaves nodes in $\{1, \dots, N\}$ that are unclustered.
        \item Take the remaining unclustered nodes and construct a $d$-regular graph, where $d$ is the smaller of the cluster size or the number of remaining nodes. 
        \item In case $\alpha_i=0$, generate a $d$-regular graph over all nodes where $d = N/\gamma$, for some $\gamma$.
        \item  Construct $(1-\eta)\lfloor N/(\alpha_i+\beta) \rfloor $ edges between the clusters and the remaining nodes each. 
    \end{enumerate}
    \item Consider two network states $\smat W^{(s_i)}$ and $\smat W^{(s_{i+1})}$ with respect to $(\alpha_i, s_i)$ and $(\alpha_{i+1}, s_{i+1})$ respectively and edge sets $\sset E_i$ and $\sset E_{i+1}$ respectively constructed using step 1. The transient states $\smat W^{(s_i+1)}, \dots, \smat W^{(s_{i+1}-1)}$ are constructed by progressively adding/deleting relevant edges $(\sset E_i \cup \sset E_{i+1}) \setminus (\sset E_{i} \cap \sset E_{i+1}) $.  
    \item Repeat step 1 and 2 over all pairs of states $(\alpha_i, s_i)$ and $(\alpha_{i+1}, s_{i+1})$ to generate the  network $\{\smat W^{(t)}\}_{t=1, \dots, T}$.
\end{enumerate}

The role of the parameters $\eta, \gamma$ and $\beta$ are as follows:
\begin{itemize}
    \item The quality parameter $\eta > 0$ is related to an individual cluster and represents the maximal ratio of number of intracluster edges to the total number of edges of cluster. 
    \item The density parameter $\gamma \in \mathbb N$ represents the inverse fraction of edges, compared to a complete network, that a $d$-regular network representing the unclustered vertices will have.
    % of the number of edges between vertices in a regular network that represents the network of unclustered vertices. 
    For higher values of $\gamma$, the regular network has fewer edges between vertices.
    \item The parameter $\beta > 0$ represents the ratio of number of vertices assigned as clustered to those assigned as unclustered. From step 1-a, note also that higher values of $\beta$ result in smaller sized clusters. 
\end{itemize} 
Higher values of  $\beta$ and $\gamma$ result in fewer constructed edges per regular graph, which results in a good balance between fast and smooth transitions in time. In our experiments in Sec.~\ref{sec:ex1} and Sec.~\ref{sec:ex2} we chose  $\alpha = \{0,1\}, s = \{ 
1,21 \}$ and $\alpha = \{0,1,2\}, s=\{1,40,60\}$ respectively. For both experiments we fix $\beta = 1.5$, $\gamma=3$ and $\eta=0.8$. 

\section{Matching time-slice clusters from Leiden}\label{sec:app:matching}

%\gf{GF to add some preliminary text here to describe our matching algorithm.}

The application of the Leiden algorithm on each time slice yields a sequence of partitions $\{\spartn X^{(1)},\ldots,\spartn X^{(T)}\}$, where the partition on time slice $t$ is $\spartn X^{(t)}=\{\sset X^{(t)}_1,\ldots,\sset X^{(t)}_{K_t}\}$ and $\sset X^{(t)}_k\subset \sset V$, for $t=1,\ldots,T$ and $k=1,\ldots,K_t$.
We wish to combine these clusters across time to form a single spacetime partition $\partn{X}=\{\set{X}_1,\ldots,\set{X}_K\}$, where $K=\max_{1\le t\le T}K_t$.
A spacetime partition element $\set{X}_k\subset \set{V}$, $1\le k\le K$ will be of the form $\smash{ \set{X}_k=\cup_{t=1}^T {\{t\} \times }\sset X^{(t)}_k }$, but it is nontrivial to sort the indices of the elements of $\spartn X^{(t)}$ for each $t=1,\ldots,T$ in such a way that this union creates sensible spacetime clusters $\set{X}_k$ with small spacetime cut values $\func \sigma(\set{X}_k)$, defined in~\eqref{eq:cut}. 
This section describes how to correctly perform this reindexing on each time slice in order to link the partitions $\spartn X^{(t)}$ across time to create a spacetime partition $\partn{X}$ with a low total spacetime cut value $\func \sigma(\partn{X})=\max_{1\le k\le K}\func \sigma(\set{X}_k)$.

\subsection{Matching clusters between adjacent time slices}
\label{sec:matching}
Fixing $t$, we consider the partitions $\spartn X^{(t)}$ and $\spartn X^{(t+1)}$ and wish to match partition elements at time $t$ to those at~$t+1$.
Each partition has a given indexing of its elements.
Without loss, we suppose that $K_t\ge K_{t+1}$ (if not, it will be clear how to match from $t+1$ to~$t$).
We compute cut values between all pairs of clusters, indexed by $1\le i\le K_t$ and $1\le j\le K_{t+1}$, leading to a $K_t\times K_{t+1}$ matrix $\smat C$ containing these values.
\begin{equation}
\label{eq:cutweights}
    \smat C_{ij}:= -a^2 \sum_{\substack{x \in X^{(t)}_i \\ y \in X^{(t+1)}_j}}  \adjtemp_{(t,x),(t+1,y)}.
\end{equation}

We consider an auxiliary bipartite graph with two sets of vertices: $\{u_1,\ldots,u_{K_t}\}$ and $\{w_1,\ldots,w_{K_{t+1}}\}$.
The first set of vertices represent the $K_t$ clusters at time $t$ and the second set of vertices represent the $K_{t+1}$ clusters at time~$t+1$.
Linking cluster $i$ at time $t$ with cluster $j$ at time $t+1$ will be indicated by including an edge (in an edge cover, see below) between $u_i$ and~$w_j$.
For all pairs $(i,j)$ the edge from $u_i$ to $w_j$ is given a weight equal to the cut value $\smat C_{ij}$ given by~\eqref{eq:cutweights}.
We wish to find a restricted maximum weighted edge cover for the bipartite graph $\sset G$ with vertices $\{u_1,\ldots,u_{K_t},w_1,\ldots,w_{K_{t+1}}\}$ and weights $\smat C_{ij}$, $i=1,\ldots,K_t$, $j=1,\ldots,K_{t+1}$.
Requiring an edge cover means that all nodes have at least one incident edge. 
We restrict the cover by insisting that each vertex in 
$\{u_1,\ldots,u_{K_T}\}$ has exactly one incident edge.
%while each vertex in $\{w_1,\ldots,w_{K_{T-1}}\}$ will have at least one incident edge\gf{GF to check this sentence after precise arc weights are included.}.
If the cover includes an edge linking $u_i$ and $w_j$ then we will connect the cluster $\smash{\sset X_i^{(t)}}$ with the cluster $\smash{\sset X_j^{(t+1)}}$; this is elaborated below.

By the above properties we see that each cluster $\smash{\sset X_i^{(t)}}$ is linked to one or more clusters $\smash{\sset X_j^{(t+1)}}$;  if more than one cluster, we merge two or more clusters from time slice $t$ into a single cluster at time slice~$t+1$. Note that in the case of appearance/disappearance of clusters we assume that the set of unclustered vertices at time $t$ is contained in $\smash{\sset X_j^{(t)}}$ for some~$j$.

The maximality of the restricted maximum weight edge covering in \eqref{eq:cutweights} means that the total intercluster weight is minimised. 
We solve a restricted maximum weight edge cover problem (RMWECP) as a binary integer program, with binary variables $\smat A_{ij}\in\{0,1\}$, $i=1,\ldots,K_t; j=1,\ldots,K_{t+1}$, and where $\smat A_{ij}=1$ if an edge in the cover is incident on $u_i$ and $w_j$ ($\smat A_{ij}=0$ otherwise):
\begin{eqnarray*}
\max_{\smat A_{ij}}&&\sum_{i=1}^{K_t}\sum_{j=1}^{K_{t+1}}\smat C_{ij}\smat A_{ij}\\
\mbox{subject to}&&\sum_{i=1}^{K_t}\smat A_{ij}=1\quad\mbox{for }j=1,\ldots,K_{t+1}\\
&&\sum_{j=1}^{K_{t+1}}\smat A_{ij}\ge 1\quad\mbox{for } i=1,\ldots,K_t\\
&&\smat A_{ij}\in\{0,1\}\quad\mbox{for } i=1,\ldots,K_t; j=1,\ldots,K_{t+1}
\end{eqnarray*}

Julia code (\cite{bezanson2017julia}) to compute the RMWECP for bipartite graphs is given below (using the JuMP modelling language (\cite{Lubin2023}) and the HiGHS solver (\cite{highs}).
{\footnotesize
%\begin{minted}{julia}
\begin{lstlisting}[language=Julia]
using JuMP, HiGHS

"""Computes a restricted maximum-weight edge cover of a bipartite graph. 
The `(i,j)th` entry of the `mxn` matrix `C` contains the weight of the 
edge joining node `i` and node `j`. 
We assume `m` is not less than `n` and restrict to unit node degrees
on the `m` nodes.
A binary `mxn` array `A` is output, encoding the cover."""
function bipartite_maximum_weight_edge_cover(C)

    # example call 
    # C = [2.5 2.5 2.5; 0 0 1; 1 1 1; 2 2 0]  
    # bipartite_maximum_weight_edge_cover(C)

    # ensure more rows than columns in C
    m, n = size(C)
    if m < n
        m, n, C = n, m, C'
    end

    # create model
    edgecover = Model(HiGHS.Optimizer)
    set_silent(edgecover)

    # define binary variable array
    @variable(edgecover, A[1:m, 1:n], Bin)

    # one cluster from the large group must be assigned to 
    # one cluster in the small group
    @constraint(edgecover, [i = 1:m], sum(A[i, :]) == 1)
    # each cluster in the small group is assigned to 
    # at least one cluster in the large group
    @constraint(edgecover, [j = 1:n], sum(A[:, j]) >= 1)

    # maximise total weight of the cover
    @objective(edgecover, Max, sum(C .* A))
    optimize!(edgecover)
    println("Total cover weight is ", objective_value(edgecover))

    return Bool.(value.(A))

end
\end{lstlisting}
%\end{minted}
}

\subsection{Matching clusters across all time slices}
Suppose that we have combined the partitions $\spartn X^{(1)},\ldots, \spartn X^{(t)}$ into a partial spacetime partition $\partn{X}$ (i.e.\ up to time $t$).
For simplicity of exposition, we assume that the indexing for $\partn{X}$ matches that of $\spartn X^{(t)}$.
To extend this to time $t+1$ we find a minimum weight edge cover for the graph generated as above using $\spartn X^{(t)}$ and $\spartn X^{(t+1)}$.
Having found this minimum weight edge cover, we combine $\spartn X_i^{(t)}$ and $\spartn X_j^{(t+1)}$ into a single spacetime partition element if the edge joining $u_i$ and $w_j$ is in the cover.
If $K_{t+1}<K_t$ this will mean that some clusters are merged between slice $t$ and slice $t+1$. If $K_{t+1}=K_t$, we retain the same number of clusters from $t$ to $t+1$ and index in such a way to naturally evolve the clusters in $\spartn X^{(t)}$ into those in $\spartn X^{(t+1)}$. If $K_{t+1} > K_t$ we perform the matching after reversing time, which leads to a splitting of one or more clusters from time $t$ to time $t+1$.

\begin{figure}[t]
    \centering
    \includegraphics[width=\textwidth]{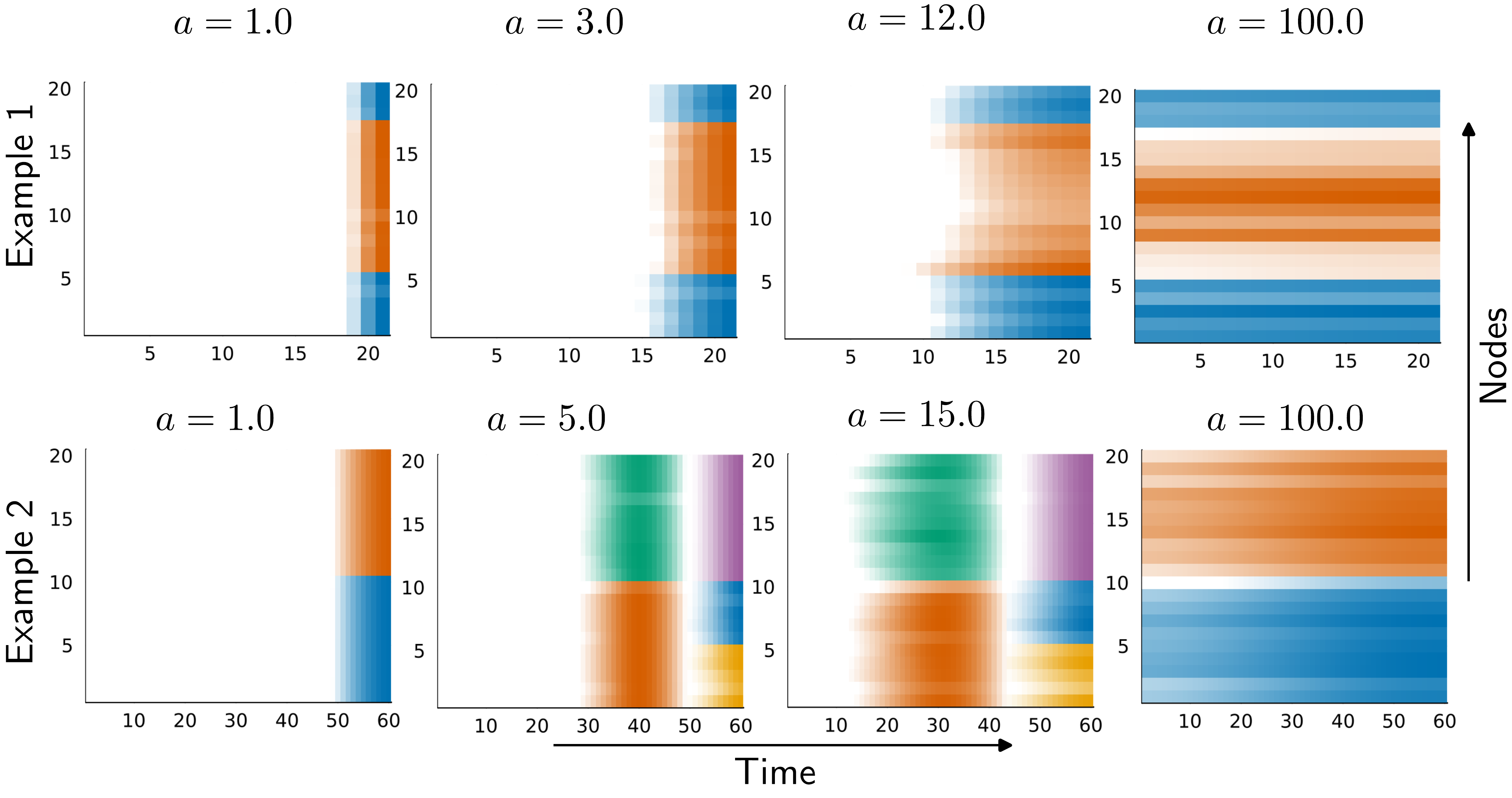}
    \caption{Cluster assignment using the spectral partitioning Algorithm~\ref{alg:spec_part} in Examples 1 and 2 for different values of the diffusion parameter $a$. Different colours correspond to different clusters. For each value of $a$, we perform spectral partitioning with the first two spatial eigenvectors in Example 1, and the first three spatial eigenvectors in Example 2. The middle two values of $a$ for each example are roughly the extremal values for which we see results consistent with those shown in Sec.~\ref{sec:ex1} and \ref{sec:ex2}.}\label{fig:varyinga}
    
\end{figure}

\section{Varying $a$ in Examples 1 and 2}\label{sec:app:varyinga}

In this section we comment on the effect of varying the temporal diffusion parameter $a$ on cluster assignment. According to Algorithm~\ref{alg:spec_part} the choices of $a$ for Examples 1 and 2 are $a=8.16$ and $a=9.27$ respectively. Here we vary $a$ and perform cluster assignment using Algorithm~\ref{alg:spec_part} as before, except for the implementation of step 3. In Fig.~\ref{fig:varyinga} we show the resulting clusters for four values of $a$, for both examples 1 and 2. In both cases, we show the effect of $a \ll a_c$ and $a \gg a_c$ by applying Algorithm~\ref{alg:spec_part} with $a=1$ and $a=100$ respectively. The other two cases {use} empirically obtained $a$ values for which we obtain cluster assignments qualitatively consistent with those obtained using $a=a_c$.

In Fig.~\ref{fig:varyinga} (upper) we show the clusters obtained for Example 1 using $a=1.0, 3.0, 12.0$ and $a=100.0$. In all cases, we set $R=1$ in Algorithm~\ref{alg:spec_part} as there is a significant spectral gap between the second and third spatial eigenvalues. SEBA produces two vectors from the leading nontrivial spatial eigenvectors which give the two cluster picture that we see for all values of $a$. We note that as $a$ increases, the transition point from $\Omega$ to cluster appearance moves further to the left in time. This is consistent with the gradual loss of edges at a constant rate, so that there is no clear time at which one can say a cluster appears.

The dynamics in Example 2 are more involved, and are shown in Fig.~\ref{fig:varyinga} (lower). For the extreme cases $a=1.0$ and $a=100.0$ we set $R=1$ due to the significant spectral gap between the first and second nontrivial spatial eigenvalues. The results for $a=5.0$ and $a=15.0$ are consistent with those shown in Fig.~\ref{fig:ex2}. Here we set $R=3$ following similar analysis as done in Sec.~\ref{sec:ex2}.
We observe that the two transition points, namely where two clusters appear from $\set\Omega$ and where two clusters appear from the cluster comprising spatial nodes $1,\ldots,10$, move to the left in time as $a$ increases, which is due to the fact that the loss of edges is gradual and larger $a$ promotes clusters with longer lifetimes, similar to Example~1.
% and observe that the two transition points a) $\Omega$ to a two-cluster emergence and b) splitting move to the left in time as $a$ increases, which is consistent with our observation in Example 1 
%\gf{I found it hard to parse this sentence}\pk{Better? The old version is commented out.}. 
We observed a sharp difference in the spacetime clustering for $a$ a little less than 5.0 and $a$ a little larger than 15.0 primarily because we applied a subjective choice based on the  Cheeger ratios $\boldsymbol{H}(\boldsymbol{X}_k)$ of sets $\boldsymbol{X}_k$ corresponding to SEBA vectors to discard some $\boldsymbol{X}_k$ for $a$ just outside $[5,15]$, while for $a \in [5,15]$ these sets were not discarded.
The SEBA vectors themselves have similar supports to the SEBA vectors for $5\leq a \leq 15$. 
%clustering  The reason for the sharp difference between the $a=1$ and $a=100$ results is primarilyWe report that for $a \lesssim 5.0$ and $a \gtrsim 15.0$ the corresponding SEBA vectors  
%\gf{This seems to be false.  $a=1$ and $a=100$ supports are very different.}\mk{I changed the inequalities to $\gtrsim$ and $\lesssim$}. 
%However the corresponding \gfadd{Cheeger ratios} $\boldsymbol{H}(\boldsymbol{X}_k)$ of \gfadd{sets $\boldsymbol{X}_k$ corresponding to SEBA vectors} that were discarded for $a \in [5,15] $ are low for $a \notin [5,15]$, resulting in clustering that is inconsistent with the results shown in Fig.~\ref{fig:ex2}.}
%\gf{The first part in this sentence seems to duplicate the earlier text for $a=1, 100$, and I could not really understand the latter half of the sentence.}\mk{Edited to reflect that for $a \notin [5,15]$ you get the same SEBA vectors but the vertical discarded vector now has lower cut value compared to the rest and thus the resulting clustering result is very different. }

\end{appendices}

\bibliographystyle{myalpha.bst}
\bibliography{refs}

\end{document}